\documentclass[10pt,a4paper]{article}
\usepackage[T1]{fontenc}
\usepackage[utf8]{inputenc}
\usepackage{lmodern,amsmath,amssymb,amsthm,mathrsfs}
\usepackage[textwidth=126mm,textheight=195mm]{geometry}
\usepackage{graphicx,booktabs,array,caption}
\usepackage{microtype,needspace,xurl}
\usepackage[numbers,sort&compress]{natbib}
\usepackage[hidelinks]{hyperref}
\usepackage[section]{placeins}
\graphicspath{{figures/}}
\newtheorem{theorem}{Theorem}
\newtheorem{proposition}{Proposition}
\newtheorem{lemma}{Lemma}
\newtheorem{corollary}{Corollary}

\makeatletter
\AddToHook{env/theorem/begin}{\@endparpenalty=10000\postdisplaypenalty=10000\interlinepenalty=10000}
\AddToHook{env/proposition/begin}{\@endparpenalty=10000\postdisplaypenalty=10000\interlinepenalty=10000}
\AddToHook{env/lemma/begin}{\@endparpenalty=10000\postdisplaypenalty=10000\interlinepenalty=10000}
\AddToHook{env/corollary/begin}{\@endparpenalty=10000\postdisplaypenalty=10000\interlinepenalty=10000}
\makeatother

\newcommand{\R}{\mathbb R}
\newcommand{\1}{\mathbf 1}
\newcommand{\Sn}{S_{n-1}}
\newcommand{\Hn}{H_n}
\newcommand{\IPR}{\operatorname{IPR}}
\newcommand{\diag}{\operatorname{diag}}
\newcommand{\osc}{\operatorname{osc}}
\newcommand{\cTV}{c_{\rm TV}}
\newcommand{\cH}{c_H}

\newcommand{\figurealt}[1]{%
 \begingroup
 \captionsetup{font=small,justification=raggedright,singlelinecheck=false}%
 \caption*{\textbf{Alt text:} #1}%
 \endgroup}
\usepackage{refcount}

\newcommand{\proofref}[1]{\par\nopagebreak[4]\noindent\textit{Proof in Appendix~\ref{#1}.}\par}

\hypersetup{pdftitle={Allocation Condensation in Redistributive Networks},
 pdfauthor={T. Ashwin Bhattathiripad and Vipin P. Veetil},
 pdfkeywords={allocation condensation, nonlinear network dynamics, mixing, heavy tails, multistability}}
\title{Allocation Condensation in\\Redistributive Networks}
\author{T. Ashwin Bhattathiripad \quad Vipin P. Veetil\\[3pt]
\small Economics Area, Indian Institute of Management Kozhikode,\\
\small Kozhikode, Kerala 673570, India}
\date{}
\begin{document}
\maketitle
\begin{abstract}
We study a fixed network on which new mass is allocated in proportion to a power of each node's stock and then redistributed. Transport retains a fraction of both old and new mass at each node and divides the rest according to fixed background shares. We show that a small amount of retention can sustain concentration on networks where complete redistribution would leave every node's share of new mass vanishing as the network grows. A node receiving most new mass may retain little relative to the network total, yet much more than redistribution alone would supply. For a power above one, the allocation rule amplifies this relative stock advantage at the next step. We identify retention rates tending to zero with network size that allow widely spread and concentrated stationary allocations to coexist on the same network. One of them stays close to the allocation under complete redistribution. Each node can also dominate a separate concentrated allocation, receiving a share of new mass tending to one. The stocks supporting these allocations attract nearby trajectories, although every stationary stock approaches the same background, with even its largest share tending to zero. These allocations continue to coexist under specified changes in transport that preserve the background. Redistribution can therefore spread the stock widely without spreading new allocations in the same way.
\end{abstract}

\begin{center}\small
Corresponding author: Vipin P. Veetil (\texttt{vipin@iimk.ac.in}).\\
T. Ashwin Bhattathiripad: \texttt{ashwinbtt@gmail.com}.
\end{center}
\clearpage
\section{Introduction}\label{sec:intro}
We study a fixed network in which each node's share of new resources depends on the amount it already holds. As resources move between nodes, they change the stocks that determine the next allocation. We ask whether increasingly complete redistribution prevents a few nodes from repeatedly receiving most of the new resources. We show that a node can keep receiving almost all new resources even when the largest stock share after redistribution tends to zero as network size increases.

We call the transported quantity mass and add a fixed proportion of its current total at each step. New mass is allocated in proportion to a power of each node's stock. Powers above one amplify relative stock differences in the allocation, a mechanism we call superlinear feedback. A fixed transport matrix then redistributes both old and new mass.

Since total mass grows, we work with shares. The injection profile records each node's share of new mass before transport, whereas the normalized stock records each node's share of total mass after transport. A stationary stock has constant shares, even though the amounts held continue to grow. We call concentration in the injection profile allocation condensation.

Our main transport family lets every node retain the same fraction of mass and redistributes the remainder according to fixed background shares. These shares are positive and may differ across nodes. With no retention, redistribution sets the stock to the background in one step, regardless of where new mass was allocated. The power rule then gives a unique stationary injection. We call this allocation the benchmark and require it to be diffuse, meaning that even its largest share tends to zero as network size increases. A diffuse background does not by itself make the benchmark diffuse, since the power rule can amplify unequal background shares.

Retention can sustain a different allocation because an amount that is small compared with the network total can still be large compared with a node's background stock. Consider a node receiving most of the new mass. Redistribution disperses almost all of that mass, but the small part retained can leave the node with much more stock than redistribution alone would supply. The power rule amplifies this relative advantage and can keep the next allocation concentrated at the same node. Near the diffuse benchmark, by contrast, each node receives only a small share of new mass. At the retention rates we identify, the retained amount is then small relative to that node's background stock, allowing a diffuse allocation to persist.

We identify a range of retention rates meeting both requirements for every fixed power above one, provided the benchmark is diffuse. In sufficiently large networks, the same transport rule supports a stationary injection close to the benchmark and a different concentrated injection for each node. Along the sequence of larger networks, the retained fraction tends to zero while each concentrated injection assigns its dominant node a share tending to one. The corresponding stocks are locally attracting under the original dynamics, so trajectories starting sufficiently close to any one of them converge to it. Yet every stationary stock, including any additional equilibria, approaches the same diffuse background. The network and retention rate stay fixed during each trajectory; the size limit compares different networks.

Coexistence also persists when transport departs from the exact retention-and-redistribution rule. The changes must preserve the background and satisfy bounds set by the retention rate and each receiving node's background share. This relative bound matters most at nodes with very small shares, where even a small absolute change can be large by comparison. The permitted transport may be nonsymmetric and nonreversible, but every node must still send a positive background-controlled share to every receiving node. Our result therefore does not cover arbitrary sparse networks.

On heavy-tailed backgrounds, some nodes have much larger shares than most others. Sufficiently strong feedback can then concentrate injection even under complete redistribution. Retention permits concentration with weaker superlinear feedback while a diffuse allocation remains locally attracting. A node need not have a large background share to dominate one of these concentrated allocations.

The retention bounds do not say how large a network must be before different initial stocks lead to distinct allocations. We iterate the original stock dynamics from the background and from stocks with extra mass at selected nodes. In one heavy-tailed sequence with nonquadratic feedback, these starts lead to different allocations as size increases. In another sequence, the approach to the limiting regime is slower and all tested starts reach the same profile at the reported sizes. Both sequences satisfy the asymptotic coexistence conditions. Finding only one outcome from the selected starts does not establish that other stationary allocations are absent.

Equal background shares allow a more explicit analysis of when concentrated stationary allocations appear. We determine the retention scale at which a single node can receive a share of new mass that stays bounded away from zero as the network grows. Concentrated states can already attract nearby initial stocks while the uniform state remains locally stable. For quadratic feedback, we classify all locally attracting stationary states away from the transition values.

The separation between concentrated injection and diffuse stock depends on how much mass remains at, or returns to, its injection node. We measure these returns over successive transport steps, giving older injections less weight because they make up a smaller fraction of the growing total. The resulting averages are the discounted return fractions. If every node's return fraction stays above a common positive bound as the network grows, concentrated stationary injection must leave a concentrated stock. In our nearly complete redistribution regime, all return fractions vanish, allowing the stock to become diffuse while injection stays concentrated.

Under symmetric transport and superlinear feedback, persistent discounted returns at a single node already guarantee a stationary injection whose largest share stays bounded away from zero as network size increases. The proof establishes existence but not local attraction.

\subsection{Related literature}
In models of conserved mass, condensation refers to a regime in which one or a few sites retain a finite fraction of the total as the system grows. Examples include zero-range processes, where hopping rates depend on the number of particles at the departure site, and related balls-in-boxes models~\cite{evanshanney2005,majumdarevanszia2005,bogaczburdajankewaclaw2007}. A related concentration of links can arise in growing networks with preferential attachment and fitness~\cite{bianconibarabasi2001}. Localization also occurs in eigenvector centrality, where a hub or a larger subgraph can dominate the node scores~\cite{martinzhangnewman2014,pastorsatorrascastellano2016}. Exchange models can also exhibit wealth condensation when exchanges are sufficiently restricted~\cite{bouchaudmezard2000}. Superlinear reinforcement produces winner-take-all behaviour in growing preferential-attachment networks and balls-in-bins processes with feedback~\cite{krapivskyrednerleyvraz2000,oliveira2009}. Those processes reinforce degree or occupancy. Here transport acts after every allocation on a fixed network, and the question is whether new allocations can remain concentrated as redistribution becomes increasingly complete.

Nonlinear rankings and random walks are closer to our setting. Akian, Gaubert and Ninove~\cite{akiangaubertninove2007} study self-validating rankings and establish conditions for uniqueness and multiplicity. Random walks whose routing depends on the current distribution of walkers can develop multiple equilibria, localized patterns and hysteresis~\cite{skardaladhikari2019,skardal2023}. Gravity-type transport models exhibit diffusion--localization transitions, hysteresis and more complicated nonlinear dynamics~\cite{koike2022,koike2024,koike2025}. In our model, feedback changes the allocation of new mass before a fixed transport matrix acts. The coexistence result concerns transport approaching complete redistribution, even as every stationary stock becomes diffuse.

The stationary equation is also related to personalized PageRank, which combines network transport with resets drawn from a prescribed distribution~\cite{gleich2015pagerank}. Aleja et al.~\cite{alejaetal2026} study the case in which this distribution must agree with the ranking it generates. In our model, an injection is stationary when the stock it produces leads the power rule to repeat the same injection. Allocation proportional to stock gives the transport matrix's stationary profile as the stationary injection, whereas powers above one allow several stationary allocations. Our uniqueness and convergence arguments draw on nonlinear Perron--Frobenius theory~\cite{lemmensnussbaum2012}.

For a uniform background and quadratic feedback, the stationary stock equation coincides with a special case of the grammar-acquisition model of Komarova, Niyogi and Nowak~\cite{komarovaniyoginowak2001}. Mitchener~\cite{mitchener2003bifurcation} analyses the bifurcations of the symmetric language dynamics and establishes convergence for that continuous-time system. Although the stationary equations agree, we check local stability under our original discrete-time dynamics. Appendix~\ref{app:mutation} gives the parameter correspondence.

\subsection{Organization of the paper}
Section~\ref{sec:model} defines the process and its stationary states. Section~\ref{sec:memory} establishes coexistence and its perturbation extension. Section~\ref{sec:heavy} studies heavy-tailed backgrounds and the finite-size implications of the coexistence window. Sections~\ref{sec:phase} and~\ref{sec:return} treat the uniform background and persistent discounted returns. Section~\ref{sec:discussion} concludes. Appendix~\ref{app:proofs} collects the proofs. The remaining appendices give supporting calculations, degree-class reductions and numerical methods.

\section{Model and normalized dynamics}\label{sec:model}
\subsection{Injection and the normalized stock}
We consider a network with $n\ge2$ nodes. The transport matrix $A\in\R_{\ge0}^{n\times n}$ acts on column vectors. Its entry $A_{ij}$ is the fraction of mass at node $j$ transferred to node $i$ in one step. We assume that $A$ is column-stochastic and primitive, so $\1^\top A=\1^\top$ and $A^k>0$ for some $k\ge1$. Column-stochasticity preserves total mass. Primitivity ensures that every node can receive mass from every other after any sufficiently large number of steps, without requiring a direct edge between every pair. Write
\begin{equation}
 \Sn=\{x\in\R_{\ge0}^n:\1^\top x=1\},\qquad
 \Hn=\{h\in\R^n:\1^\top h=0\}.
\end{equation}
Normalized stocks and injection profiles lie in the simplex $\Sn$. Perturbations that preserve their total lie in the tangent space $\Hn$. We write $\Sn^\circ$ for the relative interior of the simplex.

For a positive vector $u$ and feedback exponent $\theta\in\R$, the power rule is
\begin{equation}\label{eq:power-rule}
 P_\theta(u)_i=\frac{u_i^\theta}{\sum_k u_k^\theta},\qquad P_0(u)=\1/n.
\end{equation}
Positive exponents favour nodes with more stock, negative exponents favour nodes with less stock, and a zero exponent gives every node the same share. Superlinear feedback, $\theta>1$, amplifies relative stock differences in the allocation of new mass. For $\theta>0$, the rule extends to nonzero nonnegative vectors by $0^\theta=0$. At $\theta=0$, we use uniform injection also on the boundary. Negative exponents require strictly positive stocks. The normalized weights are also known as an escort distribution~\cite{beckschlogl1993}.

Let $m_t$ denote the nonnegative mass vector and $M_t=\1^\top m_t$ its total. We take $m_0\ne0$ and require $m_0>0$ coordinatewise when $\theta<0$. We call these initial mass vectors admissible. At each step, the forcing rate $\pi>0$ adds mass $\pi M_t$, distributed according to $P_\theta(m_t)$. Transport then acts on both the existing and the new mass. The law of motion is
\begin{equation}\label{eq:mass}
 m_{t+1}=A\{m_t+\pi M_tP_\theta(m_t)\},\qquad M_{t+1}=(1+\pi)M_t.
\end{equation}
Since the power rule is unchanged by multiplication of its argument by a positive scalar, the normalized standing stock $x_t=m_t/M_t$ obeys
\begin{equation}\label{eq:physical}
 x_{t+1}=\Phi(x_t):=A\{a x_t+bP_\theta(x_t)\},\qquad
 a=\frac1{1+\pi},\quad b=\frac\pi{1+\pi}.
\end{equation}
Here $a$ is the fraction of the post-injection total carried over from existing mass, and $b$ is the fraction newly added. The injection profile $\gamma_t=P_\theta(x_t)$ records the allocation before transport. A stationary stock is a fixed point of the normalized dynamics, so its shares remain constant while total mass continues to grow.

\subsection{Stationary injection and the discounted response}
Holding the injection profile fixed at $z\in\Sn$ replaces~\eqref{eq:physical} by the linear recursion $x_{t+1}=aAx_t+bAz$. Since $a<1$, its iterates converge to
\begin{equation}\label{eq:Q}
 Qz,\qquad Q=bA(I-aA)^{-1}=b\sum_{k=0}^\infty a^k A^{k+1}.
\end{equation}
The discounted response matrix $Q$ averages the transport of injection $z$ over paths of different lengths. The weights $ba^k$ sum to one, so $Qz$ remains a probability vector. Since $A$ is primitive, at least one term in the series is strictly positive, making $Q$ strictly positive as well. The series describes one mandatory transport step followed by a geometrically distributed number of additional steps, with mean $a/b=1/\pi$. A lower forcing rate therefore gives greater weight to longer transport paths.

Applying the power rule to this response defines the frozen-response map
\begin{equation}\label{eq:T}
 T(z)=P_\theta(Qz).
\end{equation}
A stationary injection profile $\gamma$ is a fixed point of $T$. Its stationary stock $y$ satisfies
\begin{equation}\label{eq:stationary}
 \gamma=P_\theta(Q\gamma),\qquad y=Q\gamma,\qquad \Phi(y)=y.
\end{equation}
Conversely, $\Phi(y)=y$ implies $y=QP_\theta(y)$ and hence gives a fixed point of $T$. Fixed points of $T$ and $\Phi$ are therefore in one-to-one correspondence. Their trajectories and derivatives differ, however. Solving the stationary equation for $T$ does not establish attraction under the original dynamics $\Phi$. A stationary pair exists for every real $\theta$. Appendix~\ref{app:selection} gives the existence result and sufficient conditions for uniqueness and convergence.

Let $v$ be the positive Perron probability vector of $A$. Then $Av=Qv=v$, and the benchmark injection profile is
\begin{equation}\label{eq:benchmark}
 g=P_\theta(v).
\end{equation}
For rank-one transport $A=v\1^\top$, every admissible stock reaches $v$ in one step, and $g$ is the unique stationary injection profile. For a general transport matrix, $g$ need not be a fixed point of $T$.

\subsection{Concentration and linear mixing}
For any probability vector $z$, its inverse participation ratio is $\IPR(z)=\sum_i z_i^2$. The bounds
\begin{equation}\label{eq:ipr-bounds}
 \frac1n\le\IPR(z)\le\max_i z_i,
 \qquad (\max_i z_i)^2\le\IPR(z)
\end{equation}
show that the IPR vanishes exactly when the largest share vanishes. We call a sequence diffuse when its IPR tends to zero, and localized on a few nodes when its IPR has a positive lower limit. A diffuse profile may still have an IPR that decays more slowly than $1/n$. Complete single-site condensation means that the $\ell^1$ distance to some vertex of the corresponding simplex tends to zero.

Write $e_j$ for unit mass at node $j$. For stochastic $A$, define
\begin{equation}\label{eq:TV}
 \cTV(A)=\frac12\max_{j,k}\|Ae_j-Ae_k\|_1,
 \qquad \operatorname{gap}(A)=1-\max_{\lambda\ne1}|\lambda|.
\end{equation}
The total-variation contraction coefficient $\cTV(A)$ measures how much a transport step can preserve differences between normalized stocks. The spectral gap is one minus the largest modulus of a non-Perron eigenvalue. These linear measures do not by themselves determine the basins of the nonlinear map $\Phi$.

\section{Coexistence under vanishing retention}\label{sec:memory}
\subsection{Retention, redistribution and the coexistence window}
Throughout this section, $a,b\in(0,1)$ with $a+b=1$ and $\theta>1$ are fixed independently of $n$. For a strictly positive $v_n\in\Sn$, consider
\begin{equation}\label{eq:A0}
 A_n^0=\rho_n I+(1-\rho_n)v_n\1^\top,\qquad 0<\rho_n<1.
\end{equation}
The matrix $A_n^0$ retains a fraction $\rho_n$ of mass at its origin and redistributes the remainder according to $v_n$. Since $v_n\1^\top$ is a projection, the discounted response is
\begin{equation}\label{eq:Q0}
 Q_n^0=q_n I+(1-q_n)v_n\1^\top,\qquad
 q_n=\frac{b\rho_n}{1-a\rho_n},\qquad
 \rho_n=\frac{q_n}{b+a q_n}.
\end{equation}
The stationary injection equation becomes
\begin{equation}\label{eq:memory-eq}
 \gamma=P_\theta\{(1-q_n)v_n+q_n\gamma\}.
\end{equation}
The discounted retention $q_n$ is the coefficient of the unchanged injection profile in $Q_n^0\gamma$. It includes retention over repeated transport steps, whereas $\rho_n$ measures retention in one step. Both vanish together at fixed $a,b$.

Let $Z_n$ denote the sum of the power-rule weights at $v_n$, let $v_{n,\max}$ be the largest share in $v_n$, and write $g_n$ for the benchmark injection profile,
\begin{equation}\label{eq:Zm}
 Z_n=\sum_i v_{n,i}^\theta,\qquad v_{n,\max}=\max_i v_{n,i},\qquad g_n=P_\theta(v_n).
\end{equation}
As usual, $f_n\ll h_n$ means $f_n/h_n\to0$, while $f_n\asymp h_n$ means comparison by positive constants independent of $n$. We assume that the benchmark injection profile $g_n$ is diffuse,
\begin{equation}\label{eq:deloc-hyp}
 \max_i g_{n,i}=\frac{v_{n,\max}^\theta}{Z_n}\longrightarrow0,
\end{equation}
and choose discounted retention in the window
\begin{equation}\label{eq:window}
 Z_n^{1/\theta}\ll q_n\ll\frac{Z_n}{v_{n,\max}^{\theta-1}}.
\end{equation}
Near an injection concentrated at node $j$, retention gives that node stock of order $q_n$ and hence a power-rule weight of order $q_n^\theta$. The lower bound in~\eqref{eq:window} makes this weight large relative to the total background weight $Z_n$, allowing the node to keep receiving almost all new mass.

Near the benchmark $g_n$, the retained contribution $q_ng_{n,i}$ is at most $q_nv_{n,\max}^{\theta-1}/Z_n$ times the background share $v_{n,i}$. The upper bound makes this ratio vanish, so retention changes each stock share only slightly relative to its background. The diffuse-benchmark condition makes the lower and upper requirements compatible. Diffuseness of the Perron profile $v_n$ alone would not suffice; the power rule must also leave the benchmark injection $g_n$ diffuse. Uniform, heavy-tailed and other positive backgrounds can satisfy this condition.

\begin{theorem}[Coexistence under asymptotically complete mixing]\label{thm:coexist}
Suppose~\eqref{eq:deloc-hyp} and~\eqref{eq:window} hold. For all sufficiently large $n$, the dynamics~\eqref{eq:physical} with transport~\eqref{eq:A0} has at least $n+1$ distinct locally attracting stationary stocks. Their injection profiles include a diffuse branch $\gamma_n^{\rm d}$ and a condensed branch $\gamma_n^{(j)}$ near each node, with
\begin{align}
 \sup_i\left|\frac{\gamma_{n,i}^{\rm d}}{g_{n,i}}-1\right|&\longrightarrow0,\label{eq:diffuse}\\
 \sup_{1\le j\le n}\|\gamma_n^{(j)}-e_j\|_1&\longrightarrow0.\label{eq:condensed}
\end{align}
Every stationary stock, including any additional fixed points, satisfies
\begin{equation}\label{eq:stock-TV}
 \|y_n-v_n\|_1\le2q_n\longrightarrow0,\qquad
 \IPR(y_n)\longrightarrow0.
\end{equation}
\end{theorem}
\proofref{app:proof-coexist}

At each constructed stationary stock, the induced $\ell^1$ norm of $D\Phi$ on $\Hn$ tends to zero, uniformly over the condensed branches. The transport matrices also satisfy
\begin{equation}\label{eq:mixing-main}
 \cTV(A_n^0)=\rho_n\to0,\quad \operatorname{gap}(A_n^0)=1-\rho_n\to1,
 \quad\|A_n^0-v_n\1^\top\|_{1\to1}\le2\rho_n\to0.
\end{equation}
Thus one transport step contracts differences between normalized stocks by a factor $\cTV(A_n^0)=\rho_n\to0$; this is what we mean by asymptotically complete mixing. But the retained stock at a single node can still be large relative to its background share, and the power rule amplifies this relative advantage. The window~\eqref{eq:window} is sufficient for the stated coexistence; failure to satisfy it does not rule out coexistence at other retention scales.

\subsection{Profile-preserving perturbations of transport}
We now allow transport to depart from exact retention and redistribution while preserving its Perron profile. Let
\begin{equation}\label{eq:perturb}
 A_n=A_n^0+F_n,\qquad \1^\top F_n=0,\qquad F_nv_n=0,
 \qquad |(F_n)_{ij}|\le Kq_n v_{n,i},
\end{equation}
where $K<\infty$ is independent of $n$. The weighted entrywise bound controls the change in transport relative to the background share of each receiving node, including nodes with very small shares. For uniform $v_n$, it requires each change to be at most $Kq_n/n$, a stronger restriction than an induced norm bound of order $q_n$. The perturbations may be nonsymmetric and nonreversible, but must preserve the dense redistribution background. Arbitrary small-norm perturbations need not satisfy this entrywise restriction.

\begin{theorem}[Profile-preserving robustness]\label{thm:robust}
Suppose~\eqref{eq:deloc-hyp}, \eqref{eq:window} and~\eqref{eq:perturb} hold. For all sufficiently large $n$, $A_n$ is strictly positive, and the original dynamics has at least $n+1$ distinct locally attracting stationary stocks. Their injection profiles satisfy~\eqref{eq:diffuse}--\eqref{eq:condensed}. Every stationary stock satisfies
\begin{equation}\label{eq:perturbed-stocks}
 \sup_{y_n\text{ stationary}}\|y_n-v_n\|_1\le(2+K/b)q_n\to0,
 \qquad\sup_{y_n\text{ stationary}}\IPR(y_n)\to0.
\end{equation}
\end{theorem}
\proofref{app:proof-robust}

At the constructed stationary stocks, the tangent derivative norms of $\Phi$ tend to zero, uniformly over the condensed branches. The perturbation bound also gives $\|A_n-v_n\1^\top\|_{1\to1}=O(q_n)$ and $\cTV(A_n)=O(q_n)$, and every non-Perron eigenvalue has modulus $O(q_n)$. Transport therefore still approaches complete redistribution while the distinct attracting injection profiles persist.

The stock bound controls absolute differences from the Perron profile $v_n$. It does not imply that the stock and $v_n$ have the same IPR order. Appendix~\ref{app:selection} gives a relative-coordinate condition under which those orders agree.

\subsection{Numerical experiments with a heterogeneous background}\label{sec:perturb-experiment}
We examine how quickly the asymptotic coexistence result becomes visible on a smooth nonuniform background. We first use exact retention and redistribution, then add a profile-preserving nonsymmetric perturbation. In both cases, we iterate the original stock dynamics from the background and from seeded stocks, rather than solve only the stationary equation. We use $\theta=2$, $a=b=1/2$, $q_n=2n^{-0.4}$ and $n\in\{500,2000,8000,32000\}$. Set $\varphi_i=2\pi(i-1)/n$ for $i=1,\ldots,n$, and let
\begin{equation}
 v_i=\frac{1+0.35\sin\varphi_i}{\sum_{k=1}^{n}(1+0.35\sin\varphi_k)}.
\end{equation}
For vectors $u,w$ with $v^\top u=v^\top w=0$ and maximum absolute entry one, take
\begin{equation}\label{eq:numeric-perturb}
 F_{ij}=c q_n v_i u_iw_j,\qquad c\in\{0,0.75\}.
\end{equation}
These conditions preserve the column sums and Perron profile and give $|F_{ij}|\le c q_nv_i$. Appendix~\ref{app:perturb-modes} specifies the vectors and the matrix--vector implementation. Since the background entries have order $1/n$, $Z_n\asymp n^{1-\theta}=n^{-1}$, so the retention sequence satisfies~\eqref{eq:window} asymptotically.

For each size and perturbation amplitude, we iterate~\eqref{eq:physical} from $v_n$ and from $(1-q_n)v_n+q_ne_j$ at three specified nodes. Figure~\ref{fig:perturbation} shows that the background-start injection becomes less concentrated as size increases, while the three seeded outcomes remain concentrated. Every seeded outcome has its largest injection share at the seeded node. Pairwise injection distances distinguish all four outcomes at each size and perturbation amplitude.

The computed upper bounds on the stock map's tangent derivative norm are below one at every reported stock, supporting local attraction under the original dynamics. Appendix~\ref{app:perturb-outcomes} gives the residuals, relative-coordinate errors and derivative bounds, together with a two-mode perturbation experiment that also reaches four distinct profiles at each size.

\begin{figure}[!htbp]
\centering
\includegraphics[width=\linewidth]{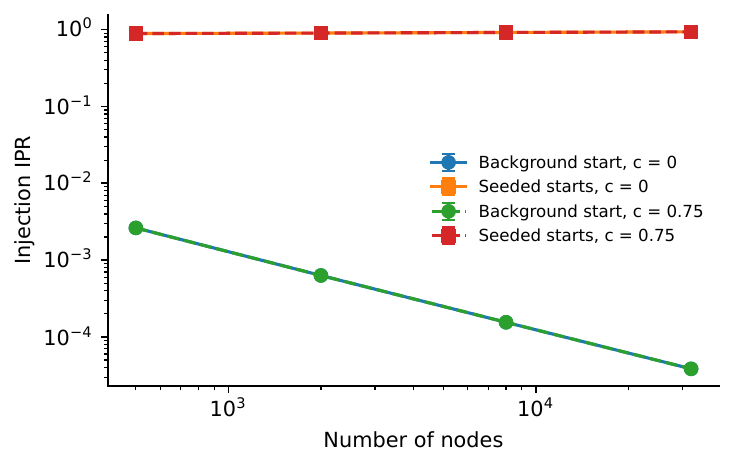}
\caption{Injection IPR for the smooth background under~\eqref{eq:numeric-perturb}, with $c=0$ and $c=0.75$. The original dynamics is iterated from the background and three seeded stocks at each size. Seeded markers give the mean over the three resulting injection profiles, and bars give their range. The runs are deterministic. The curves for $c=0$ and $c=0.75$ are nearly coincident at the plotted scale.}\label{fig:perturbation}
\figurealt{Background-start injection concentration falls with network size, whereas seeded concentration remains close to one. The two perturbation amplitudes give nearly coincident curves. Markers distinguish the initial conditions.}
\end{figure}

\section{Condensation on heavy-tailed backgrounds}\label{sec:heavy}
A heavy-tailed background can produce concentrated injection even under complete redistribution. We ask how retention changes the feedback strength needed to sustain concentration, and whether concentrated allocations can coexist with a diffuse one below the benchmark's localization threshold.

\subsection{Condensation below the benchmark threshold}
Consider the quantile profile
\begin{equation}\label{eq:quantile}
 v_{n,i}=\frac{i^{-1/\tau}}{\sum_{k=1}^n k^{-1/\tau}},\qquad \tau\in(1,2).
\end{equation}
Applying the power rule gives benchmark shares proportional to $i^{-\theta/\tau}$. For $0<\theta\le\tau$, their normalizing sum diverges and the largest benchmark share vanishes. For $\theta>\tau$, the sum converges, and the benchmark localizes. Appendix~\ref{app:tails} gives all IPR orders, including the logarithmic boundary cases and the effect of a smaller cutoff.

Under complete redistribution, $g_n=P_\theta(v_n)$ is the unique stationary injection even when $\theta>\tau$ and it is localized. Uniqueness therefore does not imply diffuse injection. Without retention the benchmark localizes only for $\theta>\tau$; at the retention rates in Corollary~\ref{cor:heavy-coexist}, coexistence with condensed allocations is already possible for any $1<\theta<\tau$.

\begin{corollary}[Coexistence below the benchmark localization threshold]\label{cor:heavy-coexist}
Let $a,b\in(0,1)$ with $a+b=1$, let $1<\theta<\tau<2$, and take $v_n$ from~\eqref{eq:quantile}. If
\begin{equation}\label{eq:sigma}
 \frac{\theta-1}{\tau}<\sigma<\frac{\theta-1}{\theta},\qquad q_n=n^{-\sigma},
\end{equation}
the original dynamics with~\eqref{eq:A0} has, for all sufficiently large $n$, a locally attracting diffuse branch with the benchmark's IPR order and $n$ locally attracting injection profiles whose IPRs tend to one. Every stationary stock has vanishing IPR. These conclusions also hold for perturbations satisfying~\eqref{eq:perturb}.
\end{corollary}
\proofref{app:proof-heavy}

For $1<\theta<\tau<2$, the diffuse-branch IPR has order $n^{2\theta/\tau-2}$, whereas the condensed-branch IPRs tend to one. The same transport matrix thus supports stationary injections with different concentration orders. Its Perron profile and tail exponent alone do not identify which concentration the dynamics selects.

\subsection{Finite-size implications of the retention window}\label{sec:finite-size}
The asymptotic window does not say how large a network must be before its attracting profiles are well separated. To examine this issue, suppress the index $n$ and define
\begin{equation}\label{eq:finite-size-parameters}
 \varepsilon=\frac{Z}{q^\theta},\qquad
 \kappa=\frac{qv_{\max}^{\theta-1}}{Z},\qquad
 \delta=\left(\frac{v_{\max}}{q}\right)^{\theta-1}=\varepsilon\kappa.
\end{equation}
The coexistence proof requires both $\varepsilon$ and $\kappa$ to become small. Increasing $q$ reduces $\varepsilon$, so a concentrated injection is easier to sustain. But the same increase raises $\kappa$, making the retained stock larger relative to the background near the diffuse branch. We choose the retention that balances these requirements,
\begin{equation}\label{eq:balanced-memory}
 q_{\rm bal}=\left(\frac{Z^2}{v_{\max}^{\theta-1}}\right)^{1/(\theta+1)},\qquad
 \eta_{\rm bal}=\left(\frac{v_{\max}^\theta}{Z}\right)^{(\theta-1)/(\theta+1)}.
\end{equation}
For a strictly positive profile with $n\ge2$ and $\theta>1$, $q_{\rm bal}$ lies in $(0,1)$ and satisfies
\begin{equation}\label{eq:balanced-diagnostic}
 \min_{0<q<1}\max\{Z/q^\theta,qv_{\max}^{\theta-1}/Z\}=\eta_{\rm bal}.
\end{equation}
At the minimum, $\varepsilon=\kappa=\eta_{\rm bal}$. Every fixed positive multiple $q_n=cq_{{\rm bal},n}$ lies in the coexistence window when $v_{n,\max}^\theta/Z_n\to0$. Appendix~\ref{app:proof-balanced} gives the calculation.

For the quantile profile~\eqref{eq:quantile} with $1<\theta<\tau$,
\begin{equation}\label{eq:balanced-rate}
 \eta_{{\rm bal},n}\asymp
 n^{-(1-\theta/\tau)(\theta-1)/(\theta+1)}.
\end{equation}
The exponent approaches zero as $\theta$ approaches either one or $\tau$. Even the balanced choice of retention can therefore leave both proof parameters far from zero over a wide range of sizes. For $\tau=1.5$ and $\theta=1.2$, the exponent is only $1/55$: halving a quantity with this asymptotic power-law decay requires increasing network size by a factor of $2^{55}$. The value of $\eta_{{\rm bal},n}$ measures how small the two parameters can be made together. It is not by itself a finite-size coexistence test, because the invariant-region and contraction estimates also contain constants.

\subsection{Nonquadratic feedback on a heavy-tailed background}\label{sec:heavy-experiment}
We now test whether different initial stocks reach distinct stationary injections before the asymptotic profiles are accurate. We use nonquadratic feedback on the exact retention and redistribution family. The profile is~\eqref{eq:quantile}, with $\tau=1.95$, $\theta=1.4$, $a=b=1/2$ and $q_n=1.1q_{{\rm bal},n}$. We selected this parameter sequence through exploratory calculations. The sequence lies in the asymptotic window, but the theorem does not specify the size at which its branches become distinguishable.

At $n=10^4,10^5,10^6$, we iterate~\eqref{eq:physical} from $v_n$ and from $(1-q_n)v_n+q_ne_j$ at nodes with the highest, middle-ranked and lowest background weights. Each transport step takes order-$n$ operations. We stop when successive stocks differ by less than $10^{-13}$ in $\ell^1$, then recompute the stationary residuals and compare the final injection profiles. Appendix~\ref{app:numerical} specifies the implementation and the profile comparisons.

Figure~\ref{fig:heavy-injection} shows how the outcomes separate with size. At ten thousand nodes, all four starts reach the same profile, concentrated at the highest-weight node. At one hundred thousand, the background and highest-weight starts still reach the same profile, while the middle-ranked and lowest-weight starts reach distinct seeded profiles. At one million, all four outcomes are distinct, with pairwise $\ell^1$ distances exceeding $1.05$. At this largest size, each seeded outcome has its largest injection share at the seeded node.

\begin{figure}[!htbp]
\centering
\includegraphics[width=\linewidth]{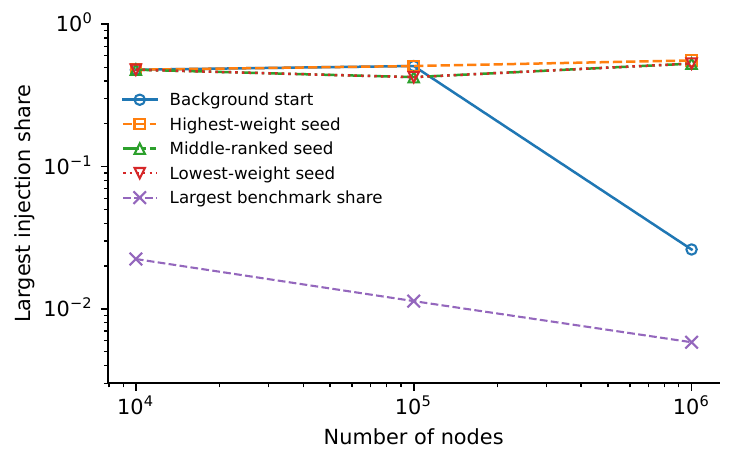}
\caption{Largest injection share reached by the original dynamics for the quantile profile with $\tau=1.95$, $\theta=1.4$ and $q_n=1.1q_{{\rm bal},n}$. Each curve follows one type of initial stock across sizes; different starts may reach the same stationary profile. The benchmark curve gives the largest share $\max_i g_{n,i}$ of $g_n=P_\theta(v_n)$.}\label{fig:heavy-injection}
\figurealt{The four outcomes coincide at the smallest size. At one million nodes, the background-start share is much lower than the seeded shares. Similar-height seeded curves represent different dominant nodes.}
\end{figure}

Table~\ref{tab:heavy-outcomes} reports the largest-size outcomes. The largest injection share is about $0.026$ for the background start and between $0.529$ and $0.558$ for the seeded starts. The computed upper bounds on the tangent derivative norm are below one for all four outcomes, supporting local attraction. The corresponding stocks are much less concentrated. Every reported stock IPR is below $7.1\times10^{-4}$ at this size, as Figure~\ref{fig:heavy-stock} shows.

% Generated by code/make_tables.py. Do not edit numerical entries by hand.
\begin{table}[!htbp]
\centering
\caption{Stationary outcomes at $n=10^6$ for $\tau=1.95$, $\theta=1.4$, $a=b=1/2$ and $q=1.1q_{\mathrm{bal}}$. The last column gives a computed upper bound on the tangent $\ell^1$ norm of $D\Phi(y)$. Original-dynamics and frozen-equation residuals are below $9.0\times10^{-14}$ and $1.9\times10^{-12}$, respectively.}\label{tab:heavy-outcomes}
\small
\begin{tabular}{lrrrr}
\toprule
Initial stock & $\max_i\gamma_i$ & $\IPR(\gamma)$ & $\IPR(y)$ & Norm bound \\
\midrule
Background & 0.026151 & 0.000864 & $8.609\times10^{-6}$ & 0.956790 \\
Highest-weight seed & 0.557269 & 0.310567 & $7.078\times10^{-4}$ & 0.771568 \\
Middle-ranked seed & 0.529556 & 0.280465 & $6.138\times10^{-4}$ & 0.847539 \\
Lowest-weight seed & 0.529547 & 0.280456 & $6.138\times10^{-4}$ & 0.847557 \\
\bottomrule
\end{tabular}
\end{table}

\begin{figure}[!htbp]
\centering
\includegraphics[width=\linewidth]{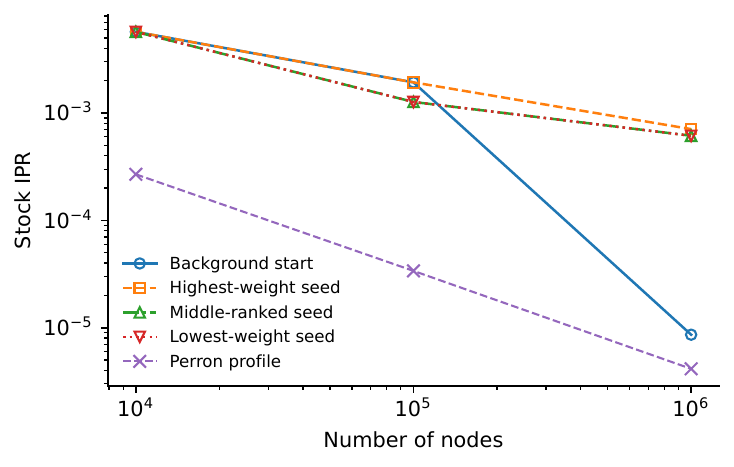}
\caption{Stationary-stock IPR for the runs in Figure~\ref{fig:heavy-injection}. At one million nodes, all four values are below $7.1\times10^{-4}$, although each seeded injection assigns more than half of the new mass to one node. The Perron-profile IPR provides a reference.}\label{fig:heavy-stock}
\figurealt{Stock concentration decreases with network size for all starts. At one million nodes, the background-start stock is closest to the Perron profile. Seeded stocks remain more concentrated.}
\end{figure}

\FloatBarrier

At one million nodes, the proof parameters are approximately $\varepsilon=0.371$, $\kappa=0.467$ and $\delta=0.173$. The background-start injection differs from the benchmark by a maximum coordinatewise relative error of about $3.49$, and the seeded shares remain well below one. The four distinct outcomes provide numerical evidence of coexisting local attractors. The benchmark approximation and complete single-site condensation are still inaccurate at this size, and seeding three nodes does not establish the theorem's conclusion uniformly over all nodes.

A slower sequence uses $\tau=1.5$, $\theta=1.2$ and $q_n=n^{-0.15}$. Its exponent also lies inside the asymptotic window, but the tested starts reach the same profile, concentrated at the highest-weight node, at every reported size up to one million. Appendix~\ref{app:heavy-scales} reports both sequences and their proof parameters. The absence of distinct outcomes in these runs does not establish that no other stationary profiles exist.

\section{Coexisting branches on a uniform background}\label{sec:phase}
We take a uniform background $v=\1/n$ and the exact retention and redistribution family~\eqref{eq:A0}. Throughout this section, $a,b\in(0,1)$ with $a+b=1$ and $\theta>1$ are fixed independently of $n$. The uniform stock is always stationary, and the multiplier of the original map $\Phi$ on every tangent direction is
\begin{equation}\label{eq:uniform-multiplier}
 \mu_{\mathrm u}=\rho(a+b\theta),\qquad
 \mu_{\mathrm u}<1\quad\Longleftrightarrow\quad q\theta<1.
\end{equation}
The equivalence follows from $\rho=q/(b+aq)$. Small discounted retention therefore leaves the uniform state locally attracting. We ask when concentrated branches appear while that state is still stable.

\subsection{Appearance of a concentrated one-site branch}
Consider an injection with share $s$ at one distinguished node and $(1-s)/(n-1)$ at each remaining node. Its stock shares are
\begin{equation}\label{eq:one-site-stocks}
 y_1=qs+\frac{1-q}{n},\qquad
 y_0=\frac{q(1-s)}{n-1}+\frac{1-q}{n}.
\end{equation}
The stationary equation reduces exactly to
\begin{equation}\label{eq:scalar-s}
 s=f_n(s):=\frac{y_1^\theta}{y_1^\theta+(n-1)y_0^\theta}.
\end{equation}
A macroscopic one-site solution has $s$ bounded away from zero as $n$ increases.

At the scale $q_n^\theta n^{\theta-1}\to C\in(0,\infty)$, the limiting stationary equation is
\begin{equation}\label{eq:limit-fold}
 s^{\theta-1}(1-s)=C^{-1}.
\end{equation}
The largest value of the left side determines the critical scale $C_c$ and the share $s_c$ at which the two limiting roots meet,
\begin{equation}\label{eq:Cc}
 s_c=\frac{\theta-1}{\theta},\qquad
 C_c=\frac{\theta^\theta}{(\theta-1)^{\theta-1}}.
\end{equation}
For $C>C_c$, write $s_-<s_c<s_+$ for the two roots of~\eqref{eq:limit-fold}.

\begin{proposition}[Vanishing-retention one-site transition]\label{prop:fold}
Fix $\theta>1$ and suppose $q_n^\theta n^{\theta-1}\to C\in(0,\infty)$. Any limit in $(0,1)$ of one-site stationary shares solves~\eqref{eq:limit-fold}. If $C>C_c$, stationary one-site branches converge to $s_-$ and $s_+$ for all sufficiently large $n$. Under the original dynamics, the upper branch is locally attracting and the lower branch is unstable, while the uniform state remains locally attracting. \end{proposition}
\proofref{app:proof-fold}

A longitudinal perturbation changes the distinguished node's stock while keeping the other stocks equal. Its multiplier tends to $\theta(1-s_\pm)$, so the perturbation grows on the lower branch and decays on the upper branch. Transverse perturbations redistribute stock among the equal background nodes; their multipliers tend to zero.

At the transition, the distinguished node already receives the positive limiting share $s_c$. The injection profile therefore differs from uniform allocation by a nonvanishing amount when the branches appear, at the fold scale $q_n\asymp n^{-(\theta-1)/\theta}$. Proposition~\ref{prop:fold} concerns these macroscopic one-site branches; it does not classify all stationary states at finite size.

\subsection{The quadratic classification}
For $\theta>1$, every nonuniform stationary stock has exactly two coordinate values, as shown in Appendix~\ref{app:proof-two-level}. When two or more nodes share the high value, transferring stock between them gives an unstable direction. A locally attracting nonuniform equilibrium must therefore have only one high-stock node.

For quadratic feedback, the remaining branches admit an exact classification. Their stationary equation is a specialization of the selection--mutation model~\cite{komarovaniyoginowak2001,mitchener2003bifurcation}. Appendix~\ref{app:mutation} gives the parameter correspondence. We evaluate stability under the original discrete-time dynamics. Let $R=y_1/y_0>1$. Cancelling the uniform factor $R-1$ in the stationary equation gives
\begin{equation}\label{eq:quadratic-R}
 (1-q)R^2-nqR+(1-q)(n-1)=0.
\end{equation}
Define
\begin{equation}\label{eq:qc}
 q_c(n)=\frac{2\sqrt{n-1}}{n+2\sqrt{n-1}},\qquad n>2.
\end{equation}
When the discriminant is positive, write
\begin{equation}\label{eq:Rpm}
 R_\pm=\frac{nq\pm\sqrt{n^2q^2-4(1-q)^2(n-1)}}{2(1-q)},
 \quad s_\pm=\frac{R_\pm^2}{R_\pm^2+n-1}.
\end{equation}
Only roots $R>1$ describe a state in which the distinguished node has the larger stock.

\begin{theorem}[All locally attracting equilibria for quadratic feedback]\label{thm:classification}
Let $n>2$, $\theta=2$ and $0<q<1$, excluding $q=q_c(n)$ and $q=1/2$.
\begin{enumerate}
\item If $q<q_c(n)$, the uniform state is the only stationary state and is locally attracting.
\item If $q_c(n)<q<1/2$, there are exactly $n+1$ locally attracting stationary states: the uniform state and the $n$ upper-branch states in~\eqref{eq:Rpm}, one for each choice of the high-stock node.
\item If $1/2<q<1$, there are exactly $n$ locally attracting stationary states, namely the same $n$ upper-branch states. The uniform state is unstable.
\end{enumerate}
Other stationary states, when present, are unstable.
\end{theorem}
\proofref{app:proof-classification}

At the excluded values $q=q_c(n)$ and $q=1/2$, some equilibria are nonhyperbolic, so the theorem leaves their local dynamics unclassified.

The fold threshold $q_c(n)$ tends to zero as $2/\sqrt n$, while the uniform state loses local stability at $q=1/2$ independently of size. Coexisting attractors therefore persist at arbitrarily small retention in large systems, as illustrated by the quadratic branches in Figure~\ref{fig:phase}. The classification concerns locally attracting fixed points; it does not assert convergence from every initial stock.

\begin{figure}[!htbp]
\centering
\includegraphics[width=\linewidth]{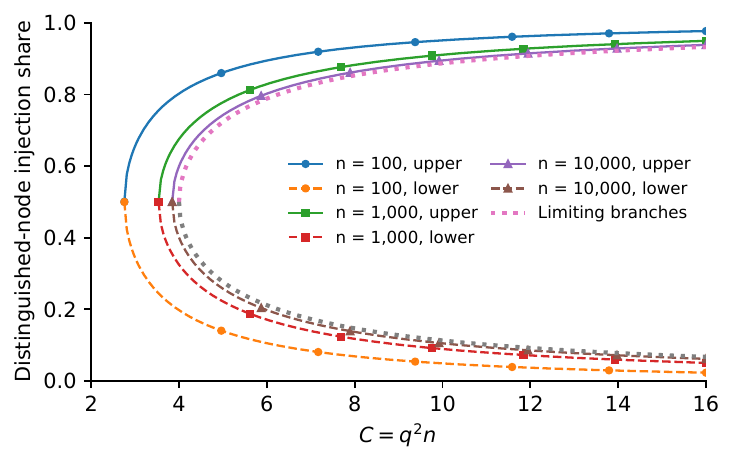}
\caption{Quadratic one-site branches from~\eqref{eq:Rpm}, plotted against $C=q^2n$. Solid and dashed curves show the upper and lower injection shares, respectively. Dotted curves give the limiting roots $(1\pm\sqrt{1-4/C})/2$ for $C\ge4$. The uniform state is locally attracting throughout the displayed range.}\label{fig:phase}
\figurealt{Upper and lower injection branches meet at a fold. As the number of nodes increases, the fold approaches a horizontal-axis value of four. Solid and dashed lines distinguish the branches; marker shapes distinguish network sizes.}
\end{figure}

\subsection{Stock--injection separation}
On the uniform background, every stationary pair satisfies
\begin{equation}\label{eq:uniform-IPR}
 \IPR(y)=\frac{1-q^2}{n}+q^2\IPR(\gamma).
\end{equation}
For quadratic feedback at the fold scale, let $q_n^2n\to C>4$. Along the upper branch,
\begin{equation}\label{eq:fold-stock-injection}
 \begin{aligned}
 s&\longrightarrow s_+(C)=\frac{1+\sqrt{1-4/C}}2,\\
 \IPR(\gamma)&\longrightarrow s_+(C)^2>0,\qquad
 \IPR(y)\sim\frac{1+C s_+(C)^2}{n}.
 \end{aligned}
\end{equation}
Thus injection remains localized while the stock has the same IPR order as a uniform profile. The limiting injection share is below one at every fixed $C$. The complete-condensation regime of Theorem~\ref{thm:coexist} lies above this fold scale.

Figure~\ref{fig:stock-injection} evaluates the exact quadratic formulas along $q=3/\sqrt n$, for which $C=9$ and the limiting largest injection share is approximately $0.873$. The scalar formulas let us extend the calculation through $n=10^7$ without simulating a full stock vector. The longitudinal multiplier stays below one. Appendix~\ref{app:uniform} gives the formula values and an exact seeding threshold on the invariant one-site line. The seeding threshold applies only when the unseeded nodes start with equal stocks. It does not determine selection from a general distribution on the simplex.

\begin{figure}[!htbp]
\centering
\includegraphics[width=\linewidth]{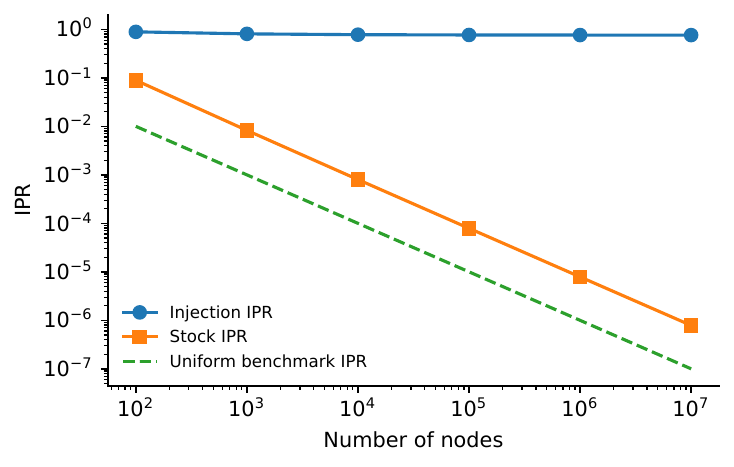}
\caption{Exact stock--injection separation along the upper quadratic branch with $q=3/\sqrt n$. Injection retains a positive limiting IPR while the stock becomes diffuse. The uniform benchmark has IPR $1/n$. All values follow from scalar stationary formulas and~\eqref{eq:uniform-IPR}, including the largest size.}\label{fig:stock-injection}
\figurealt{Injection concentration approaches a positive level as network size grows. Stock concentration falls at the same rate as the reciprocal-size benchmark, although it remains higher.}
\end{figure}

\section{Condensation sustained by persistent returns}\label{sec:return}
We now ask whether mass returning along network paths can sustain a localized stationary injection, even when the Perron profile is uniform. We consider symmetric discounted responses for which at least one node retains a positive return probability as the network grows.

\subsection{A stationary return bound}
The diagonal entry $Q_{ii}$ averages the probability that mass injected at node $i$ is at that node after one or more transport steps, with the geometric weights in~\eqref{eq:Q}. Write $r=\max_iQ_{ii}$ for the largest of these return probabilities.
\begin{theorem}[Discounted return and stationary localization]\label{thm:return}
Let $Q$ be strictly positive, symmetric and stochastic, and let $\theta>1$. There is a stationary injection $\gamma=P_\theta(Q\gamma)$ such that
\begin{equation}\label{eq:return-IPR}
 \IPR(\gamma)\ge r^{(\theta+1)/(\theta-1)}.
\end{equation}
\end{theorem}
\proofref{app:proof-return}

At fixed $\theta>1$, a return probability $r$ bounded away from zero prevents the IPR in~\eqref{eq:return-IPR} from vanishing. We obtain the stationary injection by maximizing a quadratic form under a power-norm constraint~\cite{kangnikiforov2014}. This construction proves existence. Attraction under the original dynamics $\Phi$ would require a separate stability argument.

\subsection{Regular networks and stock concentration}
For the simple random walk on a connected non-bipartite $d$-regular undirected graph, $A$ is symmetric stochastic and
\begin{equation}\label{eq:regular-return}
 Q_{ii}=b\sum_{k\ge0}a^k(A^{k+1})_{ii}
       \ge ab(A^2)_{ii}=\frac{ab}{d}.
\end{equation}
At fixed $a,b$ and bounded $d$, Theorem~\ref{thm:return} therefore yields a localized stationary injection for every fixed $\theta>1$, although the Perron benchmark is uniform. Returns along network paths can support a localized stationary injection even on this uniform background. If the walk has a holding probability $h_{\rm hold}>0$, the diagonal contribution gives
\begin{equation}\label{eq:holding-return}
 A\ge h_{\rm hold}I\quad\Longrightarrow\quad
 Q\ge\frac{bh_{\rm hold}}{1-ah_{\rm hold}}I
\end{equation}
entrywise. A positive holding probability also removes periodicity on a connected undirected graph.

\begin{proposition}[Persistent return limits stock--injection separation]\label{prop:return-stock}
If $r_0\ge0$ and $Q\ge r_0 I$ entrywise, every stationary pair obeys
\begin{equation}\label{eq:return-stock}
 \IPR(y)\ge r_0^2\IPR(\gamma).
\end{equation}
If $Q$ is also symmetric stochastic, $\IPR(y)\le\IPR(\gamma)$.
\end{proposition}
\proofref{app:proof-return-stock}

The existence result needs only one node with a large discounted return probability. The stock bound instead uses a common lower bound $r_0$ on all diagonal entries of $Q$. If that lower bound stays bounded away from zero as the network grows, any localized injection must leave a localized stock. In the vanishing-retention family with a diffuse benchmark, all diagonal returns tend to zero, allowing a diffuse stock to accompany concentrated injection.

\subsection{Numerical experiments on sparse networks}
The sparse-network experiments test whether the original dynamics reaches coexisting diffuse and concentrated profiles in the persistent-return regime, where the existence theorem alone gives no attraction result. For each $n\in\{250,1000,4000\}$, we generate five connected random $20$-regular graphs with fixed seeds and set
\begin{equation}\label{eq:sparse-numeric-A}
 A=0.8I+0.2\,\frac{\mathrm{Adj}}{20},\qquad a=b=1/2,\qquad \theta=1.2.
\end{equation}
Because the holding probability is fixed, these experiments concern persistent returns. Their transport does not approach complete redistribution as size increases. The graph seeds and initial stocks are specified in Appendix~\ref{app:sparse-methods}.

For each graph, we evolve one near-uniform start and three concentrated starts. All 15 near-uniform runs converge to the uniform state within the stated tolerance, whereas all 45 concentrated starts reach concentrated profiles. The stock remains concentrated with the injection. Every concentrated-start run has a stock-to-injection IPR ratio of at least $0.469$, above the $4/9$ bound from Proposition~\ref{prop:return-stock} when $Q\ge(2/3)I$.

Appendix~\ref{app:sparse-outcomes} reports the outcome ranges and numerical stability diagnostics. The runs support coexistence of uniform and concentrated attractors on these finite networks. They do not identify the observed profiles as variational maximizers or establish convergence from arbitrary initial stocks.

\section{Conclusion}\label{sec:discussion}
Vanishing retention can sustain a diffuse stationary allocation and a separate allocation dominated by each node, even under asymptotically complete mixing. Their stocks are locally attracting, although every stationary stock approaches the same diffuse background.

The separation between allocation and stock depends on the mass that remains at, or returns to, the node where it was injected. When discounted return fractions stay above a common positive bound across nodes as network size increases, concentrated stationary injection must also leave a concentrated stock. These fractions vanish in our nearly complete redistribution regime, so observing a diffuse stock is not enough to conclude that new mass is widely allocated. Appendix~\ref{app:reductions} and Supplementary Section S3 show how the same discounted-response framework translates degree mixing into allocation bias in reduced network models.

Local attraction alone does not tell us how frequently a concentrated allocation will be reached. In the slower heavy-tailed sequence, all tested initial stocks reach the same profile at the reported sizes, even though the sequence satisfies the asymptotic coexistence conditions. Assessing the prevalence of concentration therefore requires specifying how initial stocks are drawn and which outcome each initial stock reaches. For a process subject to disturbances, one would also need to determine how often trajectories move between basins of attraction.

\FloatBarrier
\appendix
\section{Proofs}\label{app:proofs}

\subsection{Coexistence and perturbations}

\subsubsection{Proof of Theorem~\ref{thm:coexist}}\label{app:proof-coexist}

\begin{proof}
All constants in this proof may depend on the fixed $a,b,\theta$, but not on $n$ or the distinguished node.

To check that the retention window is nonempty, write $L_n=Z_n^{1/\theta}$ and $U_n=Z_n/v_{n,\max}^{\theta-1}$. The diffuse-benchmark assumption gives
\begin{equation}\label{eq:window-exists}
 \frac{L_n}{U_n}=\left(\frac{v_{n,\max}^\theta}{Z_n}\right)^{(\theta-1)/\theta}\to0.
\end{equation}
Thus $q_n=\sqrt{L_nU_n}$ is one admissible choice.

We now suppress the index $n$ and set
\begin{equation}\label{eq:small-parameters}
 \varepsilon=Z/q^\theta,\qquad \kappa=qv_{\max}^{\theta-1}/Z,\qquad
 \delta=(v_{\max}/q)^{\theta-1}.
\end{equation}
The retention window~\eqref{eq:window} gives $\varepsilon,\kappa\to0$. Since $v_{\max}^\theta\le Z$, the limit $\varepsilon\to0$ implies $v_{\max}/q\to0$ and hence $\delta\to0$. The inequality $Z\le v_{\max}^{\theta-1}$ gives $q\le\kappa\to0$. The diffuse-benchmark assumption~\eqref{eq:deloc-hyp} also gives $v_{\max}\le v_{\max}^\theta/Z\to0$.

For a positive $w$, let $t=P_\theta(w)$ and $r_i=t_i/w_i$. Differentiation gives
\begin{equation}\label{eq:DP}
 DP_\theta(w)=\theta\{\diag(r)-t r^\top\},\qquad \1^\top DP_\theta(w)=0.
\end{equation}
On $\Hn$, $Qh=qh$, so $DT(\gamma)h=qDP_\theta(Q\gamma)h$.

To construct the diffuse branch, consider the compact convex set
\begin{equation}
 D=\{\gamma\in\Sn:g_i/2\le\gamma_i\le2g_i\text{ for all }i\}.
\end{equation}
For $\gamma\in D$ and $w=(1-q)v+q\gamma$,
\begin{equation}
 1-q\le w_i/v_i\le1-q+2\kappa.
\end{equation}
Consequently $T(\gamma)_i/g_i\to1$ uniformly in $i$ and in $\gamma\in D$, so $T$ maps $D$ into itself for large $n$. On this set, $q\max_i r_i=O(\kappa)$, and~\eqref{eq:DP} gives $\|qDP_\theta(w)\|_{1\to1}\le2\theta q\max_i r_i=O(\kappa)$. The mean-value formula on $D$ and Banach's theorem yield a unique fixed point in $D$, with~\eqref{eq:diffuse}.

To construct a branch near node $j$, use the compact convex cap
\begin{equation}
 C_j=\{\gamma\in\Sn:0\le\gamma_i\le2(v_i/q)^\theta\text{ for }i\ne j\}.
\end{equation}
For $\gamma\in C_j$,
\begin{equation}
 \gamma_j\ge1-2\varepsilon,\qquad w_j\ge q(1-2\varepsilon),\qquad
 w_i\le v_i(1+2\delta)\quad(i\ne j).
\end{equation}
It follows uniformly in $j$ that
\begin{equation}
 T(\gamma)_i\le(v_i/q)^\theta
 \left(\frac{1+2\delta}{1-2\varepsilon}\right)^\theta
 \le2(v_i/q)^\theta\quad(i\ne j)
\end{equation}
for large $n$. Hence $T(C_j)\subseteq C_j$. Throughout the cap,
\begin{equation}
 1-t_j=O(\varepsilon),\qquad q\max_{i\ne j}r_i=O(\delta),\qquad qr_j=O(1).
\end{equation}
Although $r_j$ may be of order $1/q$, the tangent derivative multiplies its contribution by the small share $1-t_j$ assigned away from node $j$. To obtain the bound while preserving this factor, apply~\eqref{eq:DP} to $e_i-e_j$ for $i\ne j$,
\begin{equation}
 \|DP_\theta(w)(e_i-e_j)\|_1
 \le2\theta\{r_i+r_j(1-t_j)\}.
\end{equation}
Every $h\in\Hn$ equals $\sum_{i\ne j}h_i(e_i-e_j)$, with $\sum_{i\ne j}|h_i|\le\|h\|_1$. Therefore
\begin{equation}\label{eq:cap-derivative}
 \|qDP_\theta(w)|_{\Hn}\|_{1\to1}=O(\delta+\varepsilon)\to0.
\end{equation}
Banach's theorem yields a fixed point in each cap $C_j$, at $\ell^1$ distance at most $4\varepsilon$ from $e_j$. For large $n$, distinct caps are disjoint. They are also disjoint from $D$, since $\max_i g_i\to0$.

The contractions on $D$ and $C_j$ establish fixed points of the auxiliary map $T$. We now check attraction of the corresponding stocks under $\Phi$. For $h\in\Hn$, both $h$ and $DP_\theta(y)h$ are in $\Hn$, where $A^0$ acts as $\rho I$. Thus
\begin{equation}\label{eq:original-derivative}
 D\Phi(y)|_{\Hn}=\rho\{aI+bDP_\theta(y)\}|_{\Hn}.
\end{equation}
At each constructed fixed point,
\begin{equation}
 \|D\Phi(y)|_{\Hn}\|_{1\to1}
 \le\rho a+\frac{\rho b}{q}\|qDP_\theta(y)|_{\Hn}\|_{1\to1}\to0,
\end{equation}
since $\rho b/q=b/(b+aq)\to1$. Differentiability near the positive stock gives local attraction.

Finally every stationary stock is $y=(1-q)v+q\gamma$, giving the distance bound in~\eqref{eq:stock-TV}. Since $\max_i y_i\le v_{\max}+q\to0$, its IPR also vanishes. Column differences of $A^0$ are $\rho(e_j-e_k)$, its restriction to $\Hn$ is $\rho I$, and $\|I-v\1^\top\|_{1\to1}\le2$. These establish~\eqref{eq:mixing-main}.
\end{proof}

\subsubsection{Proof of Theorem~\ref{thm:robust}}\label{app:proof-robust}
\label{app:robust}

\begin{proof}
All constants below may depend on the fixed $a,b,\theta,K$, but not on $n$ or on the distinguished node. We suppress the index $n$. Write $P=v\1^\top$, $A^0=\rho I+(1-\rho)P$, and $A=A^0+F$. Since $|F_{ij}|\le Kqv_i$ and $q,\rho\to0$, all entries of $A$ are strictly positive for sufficiently large $n$. Together with $\1^\top F=0$ and $Fv=0$, this makes $A$ column-stochastic with $Av=v$. Moreover $\|F\|_{1\to1}\le Kq$.

Let $R_A=(I-aA)^{-1}$ and $R_0=(I-aA^0)^{-1}$. The identity $Q=(b/a)(R_A-I)$ and the resolvent identity give
\begin{equation}\label{eq:resolvent-E}
 E:=Q-Q^0=bR_AF R_0.
\end{equation}
Both resolvents are nonnegative. Because $Av=A^0v=v$ and both matrices are stochastic,
\begin{equation}
 R_Av=v/b,\qquad \1^\top R_0=\1^\top/b.
\end{equation}
It follows entrywise that
\begin{equation}\label{eq:E-bound}
 |E|\le bR_A|F|R_0
 \le bKqR_A(v\1^\top)R_0
 =\frac Kbqv\1^\top.
\end{equation}
Set $K'=K/b$. The entrywise estimate gives $\|E\|_{1\to1}\le K'q$. The identities $\1^\top E=0$ and $Ev=0$ follow because $Q$ and $Q^0$ are stochastic and preserve $v$. For any $\gamma\in\Sn$,
\begin{equation}\label{eq:perturbed-w}
 w=Q\gamma=q\gamma+(1-q)v+E\gamma,
 \qquad |(E\gamma)_i|\le K'qv_i.
\end{equation}
On the tangent space, $Qh=qh+Eh$.

Use the same compact convex set $D=\{\gamma\in\Sn:g_i/2\le\gamma_i\le2g_i\}$ and the small parameters $\varepsilon,\kappa,\delta$ defined in~\eqref{eq:small-parameters}. Equation~\eqref{eq:perturbed-w} gives $w_i/v_i=1+O(q+\kappa)$ uniformly on $D$. Hence $T(D)\subset D$ for sufficiently large $n$, and $T(\gamma)_i/g_i\to1$ uniformly.

Put $t=P_\theta(w)$ and $r_i=t_i/w_i$. The bound $w_i/v_i=1+O(q+\kappa)$ gives $q\|DP_\theta(w)\|_{1\to1}=O(\kappa)$. For each column of $E$,~\eqref{eq:DP} and~\eqref{eq:E-bound} imply
\begin{equation}\label{eq:DPE}
 \|DP_\theta(w)E\|_{1\to1}
 \le2\theta K'q\sum_i r_i v_i.
\end{equation}
On $D$, $w_i/v_i\to1$, so $\sum_i r_iv_i=\sum_i t_i(v_i/w_i)=O(1)$. Thus
\begin{equation}
 \|DT(\gamma)|_{\Hn}\|_{1\to1}
 \le\|qDP_\theta(w)|_{\Hn}\|_{1\to1}
       +\|DP_\theta(w)E\|_{1\to1}
 =O(\kappa+q)\to0.
\end{equation}
The mean-value formula on $D$ and Banach's theorem give a unique fixed point within $D$, whose coordinate ratios to $g$ converge uniformly to one.

To construct the condensed branch at node $j$, take $C_j=\{\gamma\in\Sn: \gamma_i\le2(v_i/q)^\theta\text{ for }i\ne j\}$. Every $\gamma\in C_j$ has $\gamma_j\ge1-2\varepsilon$. For sufficiently large $n$, $1-q-K'q\ge0$, so
\begin{equation}
 w_j\ge q(1-2\varepsilon),\qquad
 w_i\le v_i(1+2\delta+K'q)\quad (i\ne j).
\end{equation}
Consequently
\begin{equation}
 T(\gamma)_i\le(v_i/q)^\theta
 \left(\frac{1+2\delta+K'q}{1-2\varepsilon}\right)^\theta
 \le2(v_i/q)^\theta\quad(i\ne j),
\end{equation}
so $T$ maps each cap $C_j$ into itself. The same coordinate bounds give
\begin{equation}
 1-t_j=O(\varepsilon),\qquad
 q\max_{i\ne j}r_i=O(\delta),\qquad qr_j=O(1).
\end{equation}
The cancellation estimate~\eqref{eq:cap-derivative} therefore gives
\begin{equation}
 \|qDP_\theta(w)|_{\Hn}\|_{1\to1}=O(\delta+\varepsilon)\to0.
\end{equation}
To bound the perturbation term $DP_\theta(w)E$, use $r_jv_j=O(v_{\max}/q)$ and
\begin{equation}
 \sum_{i\ne j}r_iv_i
 =O\left(\sum_{i\ne j}\frac{v_i^\theta}{q^\theta}\right)
 =O(\varepsilon).
\end{equation}
Equation~\eqref{eq:DPE} yields $\|DP_\theta(w)E\|_{1\to1}=O(v_{\max}+q\varepsilon)\to0$. Hence $T$ is a contraction on every cap, uniformly in $j$. The fixed points exist, satisfy $\|\gamma^{(j)}-e_j\|_1\le4\varepsilon$, and are distinct from one another and from the diffuse fixed point for large $n$.

At a constructed stationary pair, local attraction under $\Phi$ requires a bound on its own derivative. Both $h$ and $DP_\theta(y)h$ lie in $\Hn$, where $A=\rho I+F$. Thus, at each constructed fixed stock,
\begin{equation}\begin{aligned}
 \|D\Phi(y)|_{\Hn}\|_{1\to1}
 &\le(\rho+Kq)\left(a+b\|DP_\theta(y)|_{\Hn}\|_{1\to1}\right)\\
 &=a(\rho+Kq)+b(\rho/q+K)
       \|qDP_\theta(y)|_{\Hn}\|_{1\to1}\longrightarrow0.
\end{aligned}\end{equation}
The derivative bound and differentiability near the positive fixed stock give local attraction. For any stationary state, whether or not it belongs to a constructed branch,
\begin{equation}
 \|y-v\|_1\le\|q(\gamma-v)\|_1+\|E\gamma\|_1
 \le(2+K')q,
 \qquad \max_i y_i\le v_{\max}+q+K'qv_{\max}\to0.
\end{equation}
Finally, $\|A-P\|_{1\to1}\le2\rho+Kq\to0$ and
\begin{equation}
 \cTV(A)=\|A|_{\Hn}\|_{1\to1}\le\rho+Kq\to0.
\end{equation}
Every non-Perron eigenvalue has modulus at most this last bound.
\end{proof}

\subsection{Heavy-tailed backgrounds}

\subsubsection{Proof of Corollary~\ref{cor:heavy-coexist}}\label{app:proof-heavy}

\begin{proof}
Integral comparison gives
\begin{equation}
 v_{n,\max}\asymp n^{1/\tau-1},\qquad Z_n\asymp n^{1-\theta},\qquad
 \frac{Z_n}{v_{n,\max}^{\theta-1}}\asymp n^{-(\theta-1)/\tau}.
\end{equation}
Moreover, $v_{n,\max}^\theta/Z_n\asymp n^{\theta/\tau-1}\to0$. Condition~\eqref{eq:sigma} is exactly the retention window~\eqref{eq:window}. Uniform coordinatewise comparison with the benchmark preserves its IPR order. The stock conclusions follow from Theorems~\ref{thm:coexist} and~\ref{thm:robust}.
\end{proof}

\subsubsection{Balanced discounted retention}\label{app:proof-balanced}

\begin{proof}
Fix a strictly positive probability vector $v$ with at least two coordinates and $\theta>1$, and write $v_{\max}=\max_i v_i$ and $Z=\sum_i v_i^\theta$. The function $Z/q^\theta$ decreases strictly with $q$, whereas $qv_{\max}^{\theta-1}/Z$ increases strictly. Their maximum is minimized where they agree, which gives~\eqref{eq:balanced-memory}. Since $Z\le v_{\max}^{\theta-1}<1$, $q_{\rm bal}^{\theta+1}\le Z<1$. Substitution gives $\eta_{\rm bal}$ and $\delta=\varepsilon\kappa$. 
At $q_n=cq_{{\rm bal},n}$ with fixed $c>0$, the ratios $\varepsilon$ and $\kappa$ are $c^{-\theta}\eta_{{\rm bal},n}$ and $c\eta_{{\rm bal},n}$, respectively. Both tend to zero whenever $v_{n,\max}^\theta/Z_n\to0$.
\end{proof}

\subsubsection{Proof of Proposition~\ref{prop:tails}}\label{app:proof-tails}

\begin{proof}
Let $S_s=\sum_i\widehat v_i^s$ for $s>0$. The layer-cake identity gives
\begin{equation}\label{eq:layer-cake}
 S_s=s\int_0^{\max_i\widehat v_i}x^{s-1}N_n(x)\,dx.
\end{equation}
At the upper end $x=c_0n^\zeta$, the lower tail-count bound is strictly positive. Since the count is an integer, it is at least one. The upper bound then precludes $1-\zeta\tau<0$, proving $\zeta\le1/\tau$.

The contribution from $x<x_0$ is $O(n)$. The lower bound on $N_n(x_0)$ ensures that at least a fixed fraction of the coordinates exceed $x_0$, so $S_s\ge c_s n$ for all sufficiently large $n$. Between $x_0$ and $c_0n^\zeta$,~\eqref{eq:tail-assumption} bounds the integral above and below by positive constants times
\begin{equation}
 n\int_{x_0}^{c_0n^\zeta}x^{s-\tau-1}\,dx.
\end{equation}
The remaining interval, if nonempty, is bounded using monotonicity of $N_n$ and the maximum cutoff. The contribution of that remaining interval is $O(n^{1+\zeta(s-\tau)})$. Therefore
\begin{equation}\label{eq:moment-orders}
 S_s\asymp
 \begin{cases}
 n,&s<\tau,\\
 n\log n,&s=\tau,\\
 n^{1+\zeta(s-\tau)},&s>\tau.
 \end{cases}
\end{equation}
The normalization of $v$ cancels from the power-rule ratio, so $\IPR(P_\theta(v))=S_{2\theta}/S_\theta^2$. Substitution gives all five cases in~\eqref{eq:five-regimes}, including both logarithmic cases.
\end{proof}

\subsubsection{Negative powers and the lower tail}\label{app:proof-negative}

\begin{proof}
For~\eqref{eq:negative-IPR}, let $v_{\min}=\min_i v_i$ and $\theta<0$. Every numerator weight satisfies $v_i^\theta\le v_{\min}^\theta$, while the denominator obeys
\begin{equation}
 \sum_i v_i^\theta\ge cn(Cv_{\min})^\theta
       =cnC^{-|\theta|}v_{\min}^\theta.
\end{equation}
Thus every benchmark share is at most $C^{|\theta|}/(cn)$. For a probability vector $p$, $\sum_i p_i^2\le\max_i p_i$, giving the stated upper bound. Cauchy--Schwarz gives $\IPR\ge1/n$.
\end{proof}

\subsection{Uniform-background dynamics}

Throughout this subsection, $v=\1/n$, $Q=Q^0$ and $A=A^0$ are the exact retention and redistribution matrices in~\eqref{eq:A0}--\eqref{eq:Q0}. The parameters $a,b$ lie in $(0,1)$ and satisfy $a+b=1$.

\subsubsection{Proof of Proposition~\ref{prop:fold}}\label{app:proof-fold}

\begin{proof}
The scaling implies $q_n\to0$ and $nq_n\to\infty$. Uniformly for $s$ in a compact subset of $(0,1)$, $y_1\sim q_ns$ and $y_0\sim n^{-1}$. Therefore $f_n$ and its first derivative converge to those of
\begin{equation}
 f(s)=\frac{Cs^\theta}{1+Cs^\theta}.
\end{equation}
The equation $s=f(s)$ gives~\eqref{eq:limit-fold}. Differentiating $s^{\theta-1}(1-s)$ gives the unique interior maximum~\eqref{eq:Cc}. For $C>C_c$ the two roots are simple. Uniform function and derivative convergence gives a nearby root of $f_n(s)-s$ for each, for all sufficiently large $n$.

At a one-site equilibrium, tangent perturbations split into the longitudinal direction and the $n-2$ directions that redistribute mass within the equal background. Direct differentiation of~\eqref{eq:physical} gives
\begin{align}
 \mu_{\mathrm L}
 &=\rho\left[a+b\theta s(1-s)
       \left(\frac1{y_1}+\frac1{(n-1)y_0}\right)\right],\label{eq:muL}\\
 \mu_{\mathrm T}
 &=\rho\left[a+b\theta\frac{(1-s)/(n-1)}{y_0}\right].\label{eq:muT}
\end{align}
Since $\rho b/q\to1$, these converge respectively to $\theta(1-s)$ and zero. The upper root satisfies $\theta(1-s_+)<1$, whereas the lower root has $\theta(1-s_-)>1$. All these multipliers are positive. The uniform-state multiplier in~\eqref{eq:uniform-multiplier} tends to zero, so the uniform state remains locally attracting.
\end{proof}

\subsubsection{Proof of Lemma~\ref{lem:two-level}}\label{app:proof-two-level}

\begin{proof}
At stationarity, each coordinate satisfies
\begin{equation}
 c y_i^\theta-y_i+\beta=0,\qquad
 c=\frac{q}{\sum_k y_k^\theta}>0,\quad \beta=\frac{1-q}{n}>0.
\end{equation}
The function on the left is strictly convex on $(0,\infty)$, and has at most two positive roots. A nonuniform stock must use both roots, denoted $\ell<h$. Let $\gamma_h=h^\theta/\sum_k y_k^\theta$ be the injection share at a node whose stock is $h$. The derivative of $cy^\theta-y+\beta$ at $h$ is strictly positive, so
\begin{equation}
 q\theta\frac{\gamma_h}{h}=c\theta h^{\theta-1}>1.
\end{equation}
If at least two coordinates equal $h$, a tangent perturbation transferring stock between two such coordinates annihilates the rank-one term in~\eqref{eq:DP}. Its multiplier under the original dynamics is
\begin{equation}
 \rho\left(a+b\theta\frac{\gamma_h}{h}\right)
 >\rho a+\frac{\rho b}{q}=1.
\end{equation}
Here $\rho b/q=1-a\rho$. This proves instability.
\end{proof}

\subsubsection{Proof of Theorem~\ref{thm:classification}}\label{app:proof-classification}

\begin{proof}
By Lemma~\ref{lem:two-level}, every nonuniform equilibrium has $k$ high coordinates and $n-k$ low coordinates for some $1\le k\le n-1$. Its ratio $R>1$ solves
\begin{equation}\label{eq:kquadratic}
 k(1-q)R^2-nqR+(n-k)(1-q)=0.
\end{equation}
Conversely, any root $R>1$ gives normalized stock values $h=R/(kR+n-k)$ and $\ell=1/(kR+n-k)$. Equation~\eqref{eq:kquadratic} is equivalent to $kh^2+(n-k)\ell^2=q(h+\ell)$. Thus the high and low coordinates have the same residual in $y-QP_2(y)$; normalization makes that common residual zero. Each admissible root therefore determines a stationary stock.

The discriminant is nonnegative only if
\begin{equation}
 q\ge\frac{2\sqrt{k(n-k)}}{n+2\sqrt{k(n-k)}}\ge q_c(n).
\end{equation}
Thus below $q_c(n)$ no nonuniform equilibrium exists. The uniform stability condition is $q<1/2$.

For $k\ge2$, Lemma~\ref{lem:two-level} proves instability. For $k=1$, the two possible ratios are~\eqref{eq:Rpm}. Their product is $n-1$, so distinct positive roots straddle $\sqrt{n-1}>1$. The quadratic in~\eqref{eq:quadratic-R} has value $n(1-2q)$ at $R=1$. For $q_c(n)<q<1/2$, this value is positive and $R_+>1$, giving
\begin{equation}
 1<R_-<\sqrt{n-1}<R_+.
\end{equation}
For $q>1/2$, the value is negative, so $R_-<1<R_+$. At a one-site stationary pair, $DT=qDP_2(y)$ on the tangent space. Its longitudinal subspace is spanned by $(1,-1/(n-1),\ldots,-1/(n-1))$, and its transverse subspace consists of vectors with first coordinate zero and remaining coordinates summing to zero. Applying~\eqref{eq:DP} to these subspaces and using~\eqref{eq:quadratic-R} gives their respective multipliers,
\begin{equation}\label{eq:dquad}
 d_{\mathrm L}=\frac{2R(R+n-1)}{(R+1)(R^2+n-1)},\qquad
 d_{\mathrm T}=\frac2{R+1}.
\end{equation}
In particular,
\begin{equation}
 d_{\mathrm L}-1
 =-\frac{(R-1)(R^2-n+1)}{(R+1)(R^2+n-1)}.
\end{equation}
On the tangent space, $DT=qDP_\theta(y)$ and $D\Phi=\rho aI+\rho bDP_\theta(y)$. Since $\rho b/q=1-\rho a$, the original-dynamics multipliers are
\begin{equation}\label{eq:affine-multiplier}
 \mu=\rho a+(1-\rho a)d.
\end{equation}
For $R_+$, both $d_{\mathrm L}$ and $d_{\mathrm T}$ lie strictly between zero and one. Their affine transforms in~\eqref{eq:affine-multiplier} therefore give original-dynamics multipliers strictly between zero and one. For $R_->1$, $d_{\mathrm L}>1$, giving an unstable direction under the original dynamics. Counting the $n$ possible distinguished nodes completes the proof.
\end{proof}

\subsubsection{Proof of Proposition~\ref{prop:basin}}\label{app:proof-basin}

\begin{proof}
The map $\varphi_n$ is strictly increasing on $[1/n,1]$, which it maps into itself. Its only fixed points in this interval are $1/n,z_-,z_+$, by~\eqref{eq:quadratic-R}. The derivatives at $z_-$ and $z_+$ are the original-dynamics longitudinal multipliers $\mu_{\mathrm L}$ in~\eqref{eq:muL}; the derivative at $1/n$ is $\mu_{\mathrm u}$ in~\eqref{eq:uniform-multiplier}. They are less than one at $1/n$ and $z_+$, and greater than one at $z_-$. Hence $\varphi_n(z)-z$ is negative on $(1/n,z_-)$, positive on $(z_-,z_+)$ and negative on $(z_+,1]$. Monotonicity prevents an iterate from crossing an adjacent fixed point. A trajectory below $z_-$ therefore decreases to $1/n$; one between $z_-$ and $z_+$ increases to $z_+$; and one above $z_+$ decreases to $z_+$. Finally, $z_-=q_ns_{-,n}+(1-q_n)/n$, $s_{-,n}\to s_-(C)$ and $nq_n\to\infty$, proving~\eqref{eq:seed-scale}.
\end{proof}

\subsubsection{Proof of Corollary~\ref{cor:limits}}\label{app:proof-limits}

\begin{proof}
Fix $q\in(0,1)$. The condensed-cap construction in Theorem~\ref{thm:coexist} still applies for uniform $v$: now $\varepsilon=n^{1-\theta}/q^\theta\to0$ and $\delta=(nq)^{1-\theta}\to0$. Invariance and contraction of the condensed caps do not require $q\to0$. The cap construction therefore gives stationary injections tending to every vertex, so $M_n(q)\to1$. Their original-dynamics derivative norms have limiting upper bound $\rho a<1$, so they are also locally attracting for large $n$. Conversely, at fixed $n$ every stationary stock satisfies $\|y-\1/n\|_1\le2q$. Continuity of $P_\theta$ near $\1/n$ implies convergence of every stationary injection to $\1/n$, uniformly over the stationary set, and thus $M_n(q)\to1/n$ as $q\downarrow0$. The iterated limits follow.
\end{proof}

\subsection{Persistent discounted returns}

\subsubsection{Proof of Theorem~\ref{thm:return}}\label{app:proof-return}

\begin{proof}
Set
\begin{equation}\label{eq:variational}
 p=1+\frac1\theta\in(1,2),\qquad
 \Lambda=\max_{\substack{u\ge0\\ \sum_i u_i^p=1}}u^\top Qu.
\end{equation}
This is the weighted-matrix form of the $p$-spectral-radius problem~\cite{kangnikiforov2014}. We prove the stronger estimate
\begin{equation}\label{eq:variational-IPR}
 \IPR(\gamma)\ge\Lambda^{p/(2-p)}\ge r^{(\theta+1)/(\theta-1)}.
\end{equation}
A maximizer exists by compactness. To see that it is strictly positive, suppose $u_j=0$. Then adding a coordinate $t>0$ increases the quadratic form by $2t(Qu)_j+O(t^2)$, whose linear coefficient is positive. Renormalization to unit $\ell^p$ norm changes the objective only by $O(t^p)=o(t)$. A boundary maximizer is therefore impossible. At an interior maximizer, the Euler--Lagrange equation has a multiplier determined by taking the inner product with $u$. Using $\sum_i u_i^p=1$ and $u^\top Qu=\Lambda$ gives
\begin{equation}\label{eq:EL}
 Qu=\Lambda u^{p-1}.
\end{equation}
Since $\theta(p-1)=1$, normalization gives $\gamma=u/\sum_i u_i$ and $P_\theta(Q\gamma)=\gamma$.

Choosing $e_j$ at a node with $Q_{jj}=r$ gives $\Lambda\ge r$. Symmetric stochastic $Q$ has Euclidean operator norm at most one, so $\|u\|_2^2\ge u^\top Qu=\Lambda$. Write $S_t=\sum_i u_i^t$. Summing~\eqref{eq:EL} gives $S_1=\Lambda S_{p-1}$. H\"older's inequality, with weights $p-1$ and $2-p$, gives
\begin{equation}
 S_1\le S_{p-1}^{p-1}S_p^{2-p}
       =(S_1/\Lambda)^{p-1},
\end{equation}
because $(p-1)^2+p(2-p)=1$ and $S_p=1$. Hence $S_1^{2-p}\le\Lambda^{-(p-1)}$ and
\begin{equation}
 \IPR(\gamma)=\frac{\|u\|_2^2}{S_1^2}
 \ge\Lambda^{1+2(p-1)/(2-p)}=\Lambda^{p/(2-p)}.
\end{equation}
The exponent $p/(2-p)$ equals $(\theta+1)/(\theta-1)$.
\end{proof}

\subsubsection{Proof of Proposition~\ref{prop:return-stock}}\label{app:proof-return-stock}

\begin{proof}
The coordinate inequality $y_i=(Q\gamma)_i\ge r_0\gamma_i$ gives the lower bound. Symmetric stochastic $Q$ has $\|Q\|_{2\to2}\le1$, which gives the upper bound.
\end{proof}

\subsection{Benchmark comparison and convergence}

\subsubsection{Proof of Proposition~\ref{prop:transfer}}\label{app:proof-transfer}

\begin{proof}
Since $y=Q\gamma$ and $v=Qv$, the diameter bound gives $d_H(y,v)\le\Delta(Q)$. The identity $d_H(P_\theta(y),P_\theta(v))=|\theta|d_H(y,v)$ gives the second inequality. For probability vectors, $d_H(u,w)\le d$ implies $e^{-d}\le u_i/w_i\le e^d$, because the smallest ratio is at most one and the largest at least one. Squaring and summing proves the IPR comparisons. Taking logarithms of the IPR bounds and dividing by $\log n$ proves the comparison of decay exponents.
\end{proof}

\subsubsection{Proof of Proposition~\ref{prop:T-contraction}}\label{app:proof-T-contraction}

\begin{proof}
The discounted series~\eqref{eq:Q} is strictly positive because some power of $A$ is positive, and is column-stochastic by summing its geometric weights. The map $T=P_\theta\circ Q$ is therefore continuous from the compact simplex to its interior for every real $\theta$. Brouwer's theorem supplies a fixed point, and~\eqref{eq:stationary} gives a stationary pair. Birkhoff's inequality and the power identity give
\begin{equation}
 d_H(Tu,Tw)\le|\theta|\cH(Q)d_H(u,w).
\end{equation}
The relative interior $\Sn^\circ$ of the simplex is complete in Hilbert's metric. The continuous image $T(\Sn)$ is also compact, bounded away from the boundary and invariant under $T$. Banach's theorem yields uniqueness and convergence. Boundary starts enter this positive image after one step. Since $\cH(Q)<1$, the contraction condition holds in particular for $|\theta|\le1$.
\end{proof}

\subsubsection{Proof of Proposition~\ref{prop:physical-contraction}}\label{app:proof-physical-contraction}

\begin{proof}
For normalized positive $x,z$, each ratio
\begin{equation}
 \frac{ax_i+bP_\theta(x)_i}{az_i+bP_\theta(z)_i}
\end{equation}
is a convex combination of $x_i/z_i$ and $P_\theta(x)_i/P_\theta(z)_i$. The log-ratio interval for each of these two normalized-vector pairs contains zero. The diameter of the union of the two intervals is therefore at most the sum of their diameters. It follows that
\begin{equation}
 d_H(ax+bP_\theta(x),az+bP_\theta(z))
 \le d_H(x,z)+d_H(P_\theta(x),P_\theta(z))
 =(1+|\theta|)d_H(x,z).
\end{equation}
Applying $A$ gives the contraction factor $k$. The image of $\Phi$ lies in the compact positive convex hull of the columns of $A$. This convex hull is invariant under $\Phi$ and complete in Hilbert's metric, so Banach's theorem applies. Both $x_1$ and $y$ lie in the projective image of $A$, yielding $d_H(x_1,y)\le\Delta(A)$. Iteration and the power identity give the bounds. When $\cH(A)=0$, $A$ is rank one and its stochastic image is a single point.
\end{proof}

\subsection{Allocation bias and numerical derivatives}

\subsubsection{Proof of Proposition~\ref{prop:block}}\label{app:proof-block}

\begin{proof}
Write $r(\chi)=(\alpha_1\chi+\alpha_0)/(\beta_1\chi+\beta_0)$ with positive coefficients. Then
\begin{equation}
 \frac{d\log r(e^u)}{du}
 =\frac{\alpha_1 e^u}{\alpha_1 e^u+\alpha_0}
  -\frac{\beta_1 e^u}{\beta_1 e^u+\beta_0},
\end{equation}
whose absolute value is at most one. Thus $u\mapsto\theta\log r(e^u)$ is a contraction of the complete real line. The function $u-\theta\log r(e^u)$ is strictly increasing and equals $-\theta\log G$ at zero, proving the sign claim. Finally,
\begin{equation}
 |u-\theta\log G|\le|\theta||u|
\end{equation}
gives both bounds by the triangle inequality.
\end{proof}

\subsubsection{Proof of Lemma~\ref{lem:DPnorm}}\label{app:proof-DPnorm}

\begin{proof}
For a matrix restricted to the zero-sum subspace, the induced $\ell^1$ norm is one half the maximum distance between the matrix columns. Indeed, a zero-sum unit vector is a difference of two nonnegative vectors of mass $1/2$, and hence a convex combination of $(e_i-e_k)/2$. If $r_i\ge r_k$, direct use of~\eqref{eq:DP} gives
\begin{equation}
 \tfrac12\|DP_\theta(w)(e_i-e_k)\|_1
 =\theta\{r_i(1-t_i)+r_kt_i\}.
\end{equation}
For a pair containing the maximizing index $j$, the expression is largest when the other index has value $r_{(2)}$. A pair not containing $j$ gives at most $\theta r_{(2)}$, whereas the maximizing pair containing $j$ gives at least $\theta r_{(2)}$. Maximizing over column pairs therefore gives~\eqref{eq:DPnorm}.
\end{proof}

\section{Uniform-background calculations}\label{app:uniform}
For the exact uniform-background family, let $v=\1/n$, $A=A^0$ and $Q=Q^0$, with $a,b\in(0,1)$ and $a+b=1$. Its stationary equation permits a comparison with selection--mutation dynamics, although stability must be checked under the original discrete-time process.

\subsection{Correspondence with selection--mutation dynamics}\label{app:mutation}
The stationary stock equation $y=QP_\theta(y)$ also gives the stationary states of the selection--mutation flow
\begin{equation}\label{eq:selection-flow}
 \dot y=Qy^{\circ\theta}-S_\theta(y)y
       =S_\theta(y)\{QP_\theta(y)-y\},\qquad
 S_\theta(y)=\sum_i y_i^\theta.
\end{equation}

Here $y^{\circ\theta}$ denotes coordinatewise powers, and the corresponding frequency-dependent fitness is $y_i^{\theta-1}$. The flow and the original map~\eqref{eq:physical} have the same stationary stocks, but generally different trajectories.

For the exact retention family, define $\mathcal S_q(x)=Q^0P_\theta(x)$. The original stock map satisfies
\begin{equation}\label{eq:relaxed-mutation}
 \Phi(x)=a\rho x+(1-a\rho)\mathcal S_q(x).
\end{equation}
The identity follows from $(1-a\rho)q=b\rho$ and $(1-a\rho)(1-q)=1-\rho$. Each update of $\Phi$ therefore averages the current stock with $\mathcal S_q(x)$. The frozen-response map $T$ reverses the order used by $\mathcal S_q$, applying the discounted response before the power rule. At a stationary pair, $Q^0|_{\Hn}=qI$ gives $D\mathcal S_q=qDP_\theta(y)=DT(\gamma)$ on $\Hn$. Thus the eigenvalues of either derivative give the original-dynamics multipliers through~\eqref{eq:affine-multiplier}.

For $\theta=2$ and $v=\1/n$, the flow is the symmetric grammar-acquisition model of Komarova, Niyogi and Nowak~\cite{komarovaniyoginowak2001} with zero cross-type and baseline fitness. Their learning fidelity is the diagonal entry $\widehat q=q+(1-q)/n$ of $Q$, rather than the discounted retention $q$. Under this conversion, the thresholds become
\begin{equation}\label{eq:prior-threshold}
 \widehat q_c=q_c+\frac{1-q_c}{n}
 =\frac{4+2(n-1)^{3/2}-3n}{(n-2)^2},\qquad
 \widehat q_{\rm u}=\frac12+\frac1{2n}.
\end{equation}
These are the fold and uniform-stability thresholds in their Section~3. Mitchener~\cite{mitchener2003bifurcation} gives the full bifurcation and convergence analysis of the symmetric continuous-time system. The population frequencies in the grammar-acquisition model correspond to our stock $y$. Our injection profile is $P_2(y)$, which can have a different concentration order. A comparison of concentration across the models must therefore use their stock and population-frequency distributions.

\subsection{Two levels and the quadratic classification}
\begin{lemma}[Two stock levels and instability with multiple high-stock nodes]\label{lem:two-level}
For every $\theta>1$ and $0<q<1$, a nonuniform stationary stock has exactly two coordinate values. If the larger value occurs at two or more nodes, that equilibrium has an original-dynamics multiplier strictly greater than one. \end{lemma}
\proofref{app:proof-two-level}

Every locally attracting nonuniform equilibrium must therefore have exactly one high-stock node. Figure~\ref{fig:fold} shows the exact fold $q_c(n)$, its $2/\sqrt n$ asymptote, and its separation from the uniform-state stability threshold $q=1/2$.

\begin{figure}[!htbp]
\centering
\includegraphics[width=\linewidth]{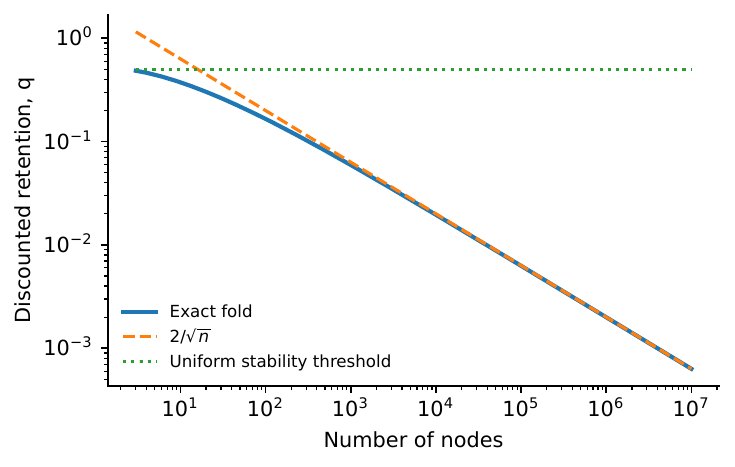}
\caption{The exact quadratic fold $q_c(n)$ and the uniform state's local-stability threshold $1/2$. The asymptote $2/\sqrt n$ is shown for comparison. Between the fold and $1/2$, the model has exactly $n+1$ locally attracting equilibria. Linearization at the uniform state gives only the upper threshold.}\label{fig:fold}
\figurealt{The fold decreases toward zero with network size, while the uniform-state stability threshold stays at one half. The interval between them contains coexisting locally attracting equilibria.}
\end{figure}

Table~\ref{tab:exact} evaluates the upper-branch formulas along $q=3/\sqrt n$, the sequence used in Figure~\ref{fig:stock-injection}. The stock IPR falls with size while the injection IPR approaches a positive limit. Both the longitudinal upper-branch multiplier and the uniform-state multiplier remain below one.

% Generated by code/make_tables.py. Do not edit numerical entries by hand.
\begin{table}[!htbp]
\centering
\caption{Exact upper-branch evaluations for $q=3/\sqrt n$, $\theta=2$ and $a=b=1/2$. Values come from the scalar stationary formulas. The stationary-equation residuals are below $10^{-12}$ in the saved size sweep. The multiplier $\mu_{\mathrm L}$ describes the longitudinal direction on the upper branch; $\mu_{\mathrm u}$ describes every tangent direction at the uniform state.}\label{tab:exact}
\small
\begin{tabular}{rrrrr}
\toprule
$n$ & $\IPR(\gamma)$ & $\IPR(y)$ & $\mu_{\mathrm L}$ & $\mu_{\mathrm u}$ \\
\midrule
$100$ & 0.888965 & $8.911\times10^{-2}$ & 0.351627 & 0.692308 \\
$1\,000$ & 0.807898 & $8.262\times10^{-3}$ & 0.286768 & 0.259944 \\
$10\,000$ & 0.776867 & $7.991\times10^{-4}$ & 0.264831 & 0.087379 \\
$100\,000$ & 0.766472 & $7.898\times10^{-5}$ & 0.257862 & 0.028193 \\
$1\,000\,000$ & 0.763125 & $7.868\times10^{-6}$ & 0.255661 & 0.008973 \\
$10\,000\,000$ & 0.762060 & $7.859\times10^{-7}$ & 0.254966 & 0.002843 \\
\bottomrule
\end{tabular}
\end{table}

\FloatBarrier
\subsection{An exact seeding threshold along the one-site line}
The uniform model preserves initial stocks in which all but one node have equal shares. For $0\le z\le1$, the line
\begin{equation}
 x(z)=\left(z,\frac{1-z}{n-1},\ldots,\frac{1-z}{n-1}\right)
\end{equation}
is invariant under the original dynamics. For $\theta=2$, its scalar stock map is
\begin{equation}\label{eq:scalar-physical}
 \varphi_n(z)=\rho\left[az+b\frac{z^2}{z^2+(1-z)^2/(n-1)}\right]
             +\frac{1-\rho}{n}.
\end{equation}
For $n>2$ and $q_c(n)<q<1/2$, the two nonuniform fixed stocks on this line have distinguished-node shares
\begin{equation}
 z_\pm=\frac{R_\pm}{R_\pm+n-1}.
\end{equation}
\Needspace{9\baselineskip}
\begin{proposition}[Basin threshold on an invariant line]\label{prop:basin}
Suppose $n>2$, $\theta=2$ and $q_c(n)<q<1/2$. For initial stocks on this line with $z_0\ge1/n$, the iterates converge to $1/n$ if $z_0<z_-$ and to $z_+$ if $z_0>z_-$. The threshold $z_-$ is itself fixed.
\end{proposition}
\proofref{app:proof-basin}

Figure~\ref{fig:basin} illustrates this threshold on the invariant one-site line with $z_0\ge1/n$: starts below $z_-$ return to uniform injection, whereas starts above $z_-$ approach the upper branch.

For $q_n^2n\to C>4$, the threshold satisfies
\begin{equation}\label{eq:seed-scale}
 z_-\sim q_n s_-(C),\qquad
 s_-(C)=\frac{1-\sqrt{1-4/C}}2.
\end{equation}

The stock increase $z_--1/n$ needed to cross the threshold vanishes, although its ratio to the uniform per-node stock $1/n$ diverges. Thus a seed can be small as a share of the total and still large relative to a node's uniform stock. The threshold applies when all unseeded nodes have equal shares; it does not measure the basin's volume in the full simplex.

\begin{figure}[!htbp]
\centering
\includegraphics[width=\linewidth]{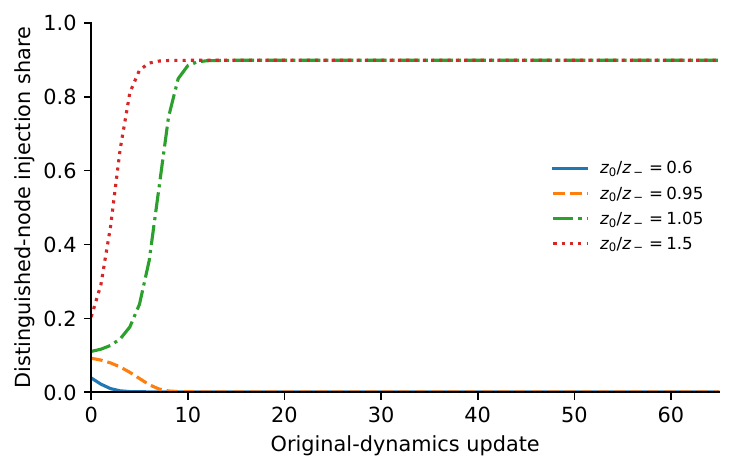}
\caption{Original stock dynamics on the invariant one-site line, with $n=1000$, $a=b=1/2$, $\theta=2$ and $q=3/\sqrt n$. Initial distinguished-node stocks are the indicated multiples of the exact threshold $z_-$. The plotted observable is the subsequent injection share. Starts below the threshold converge to uniform injection and starts above it converge to the condensed branch.}\label{fig:basin}
\figurealt{The two below-threshold starts approach uniform injection over successive updates. The two above-threshold starts approach the concentrated upper branch. Line styles identify the initial seed sizes.}
\end{figure}

\subsection{Noncommuting limits of the maximum stationary IPR}
Let $\mathcal E_{n,q}$ be the stationary injection set for the uniform background and define
\begin{equation}
 M_n(q)=\max_{\gamma\in\mathcal E_{n,q}}\IPR(\gamma).
\end{equation}
This maximum exists because the stationary set is nonempty and compact.
\begin{corollary}[Noncommuting complete-mixing and large-system limits]\label{cor:limits}
For every fixed $\theta>1$,
\begin{equation}\label{eq:noncommute}
 \lim_{q\downarrow0}\lim_{n\to\infty}M_n(q)=1,
 \qquad
 \lim_{n\to\infty}\lim_{q\downarrow0}M_n(q)=0.
\end{equation}
\end{corollary}
\proofref{app:proof-limits}

The two limits concern the largest IPR among stationary injection profiles. They do not weight those profiles by the sizes of their basins of attraction.

\section{Benchmark comparison and convergence}\label{app:selection}
Throughout this appendix, $A$ is primitive and column-stochastic, $v$ is its positive Perron probability vector, and $Q$ is the discounted response~\eqref{eq:Q}. Hilbert's projective metric compares relative coordinates. When background shares are small, a small absolute difference between stocks can still leave large relative differences.

\subsection{Transferring benchmark concentration}
For positive vectors $u,w$ and a strictly positive matrix $B$, define Hilbert's projective metric and the projective diameter by
\begin{equation}\label{eq:Hilbert-def}
 d_H(u,w)=\log\frac{\max_i u_i/w_i}{\min_i u_i/w_i},\qquad
 \Delta(B)=\log\max_{i,j,k,\ell}\frac{B_{ij}B_{k\ell}}{B_{i\ell}B_{kj}}.
\end{equation}
The diameter $\Delta(B)$ is the largest Hilbert distance between two columns of $B$, and it bounds $d_H(Bu,Bw)$ for all positive $u,w$. The associated contraction coefficient is $\cH(B)=\tanh(\Delta(B)/4)$~\cite{birkhoff1957,bushell1973}.

\begin{proposition}[Comparison with the Perron and benchmark profiles]\label{prop:transfer}
For any stationary pair and any real $\theta$,
\begin{align}
 d_H(y,v)&\le\Delta(Q),&
 d_H(\gamma,P_\theta(v))&\le|\theta|\Delta(Q),\label{eq:transferH}\\
 e^{-2\Delta(Q)}\IPR(v)&\le\IPR(y)\le e^{2\Delta(Q)}\IPR(v),\label{eq:stock-transfer}\\
 e^{-2|\theta|\Delta(Q)}\IPR(P_\theta(v))&\le\IPR(\gamma)\notag\\
 &\le e^{2|\theta|\Delta(Q)}\IPR(P_\theta(v)).\label{eq:injection-transfer}
\end{align}
\end{proposition}
\proofref{app:proof-transfer}

For fixed $\theta$, bounding $\Delta(Q_n)$ uniformly in $n$ keeps every stationary injection's IPR within constant factors of the benchmark IPR. If only $\Delta(Q_n)=o(\log n)$ is known, the stationary injection has the same limiting decay exponent as the benchmark whenever $-\log\IPR(g_n)/\log n$ converges. This weaker condition allows subpolynomial differences in IPR, so it does not preserve constant-factor orders or establish a nonvanishing IPR.

The uniform bound on projective diameter can fail on sparse networks. For example, on an undirected graph with degrees $d_i$ and transition $A_{ij}=1/d_j$ along edges, an edge from $j$ to $i$ gives
\begin{equation}\label{eq:sparse-obstruction}
 \Delta(Q)\ge\log\frac{b\sum_k d_k}{d_i d_j}.
\end{equation}
Indeed $Q_{ij}\ge b/d_j$, $v_i=d_i/\sum_kd_k$, and $d_H(Qe_j,v)\ge\log(Q_{ij}/v_i)$. The lower bound on Hilbert distance follows because the $v$-weighted average of the coordinate ratios $(Qe_j)_k/v_k$ is one. Approximating $e_j$ from the interior justifies using the image-diameter bound. If the graph has order-$n$ total degree and an edge with $d_id_j=o(n)$, the diameter diverges. For the uniform retention family, the diameter can be calculated directly as
\begin{equation}\label{eq:Hilbert-reset}
 \Delta(Q_n^0)=2\log\left(1+\frac{nq_n}{1-q_n}\right).
\end{equation}
For fixed $\theta>1$ and $q_n\asymp n^{-(\theta-1)/\theta}$, $nq_n\to\infty$, so the projective diameter diverges while $\cTV(Q_n^0)=q_n\to0$. Transport can therefore remove almost all absolute differences between probability vectors while leaving large differences relative to their small coordinates.

\subsection{Frozen-response contraction and stock convergence}
\begin{proposition}[Frozen-response contraction]\label{prop:T-contraction}
For every primitive column-stochastic $A$ and real $\theta$, a stationary pair exists. If
\begin{equation}\label{eq:T-contraction}
 |\theta|\cH(Q)<1,
\end{equation}
the stationary pair is unique, and the iterates of $T$ from every initial injection profile in the simplex converge geometrically to its injection profile in Hilbert distance after the first iterate.
\end{proposition}
\proofref{app:proof-T-contraction}

Since $Q>0$ gives $\cH(Q)<1$, condition~\eqref{eq:T-contraction} holds for every $|\theta|\le1$. Stronger projective contraction allows uniqueness at larger feedback magnitudes.

A unique stationary injection can still be localized. At rank-one transport, $A=v\1^\top$, the stock reaches $v$ in one step and the unique injection is $P_\theta(v)$ for every real $\theta$. With~\eqref{eq:quantile} and $\theta>\tau$,
\begin{equation}\label{eq:unique-localized}
 \IPR(P_\theta(v_n))\longrightarrow
 \frac{\zeta_{\rm R}(2\theta/\tau)}{\zeta_{\rm R}(\theta/\tau)^2}>0,
\end{equation}
where $\zeta_{\rm R}$ is the Riemann zeta function.

The original stock process retains part of the current stock before transport, so convergence of the auxiliary iteration $T$ does not settle convergence of $\Phi$. A stronger mixing condition on $A$ gives a direct bound for the stock process.
\begin{proposition}[A sufficient condition for global stock convergence]\label{prop:physical-contraction}
Suppose $A>0$ and
\begin{equation}\label{eq:physical-contraction}
 k:=(1+|\theta|)\cH(A)<1.
\end{equation}
Then the original map $\Phi$ has a unique fixed stock and every admissible trajectory converges to it geometrically in Hilbert distance. For $t\ge1$,
\begin{equation}
 d_H(x_t,y)\le k^{t-1}\Delta(A),\qquad
 d_H(P_\theta(x_t),\gamma)\le|\theta|k^{t-1}\Delta(A).
\end{equation}
\end{proposition}
\proofref{app:proof-physical-contraction}

For $k=0$, the bounds are read with $k^0=1$ at $t=1$, and the stock reaches its fixed point after one transport step.

The contraction conditions are sufficient, not necessary; failure of either test does not establish multiplicity or instability. At $\theta=1$, the original map is exactly $x\mapsto Ax$, which converges for every primitive $A$ even when~\eqref{eq:physical-contraction} fails. The assumption $A>0$ is also stronger than primitivity. Adding holding probabilities to a sparse walk removes periodicity without making all matrix entries positive.

Supplementary Material, Section S1 gives finite-block and stable-filter convergence criteria for primitive transport, including sparse matrices, with explicit geometric rates. Section S2 separates these trajectory bounds from approximation over the discounted transport horizon.

\subsection{Local derivatives of the two maps}
For a doubly stochastic $A$, the uniform state $e=\1/n$ is stationary and $DP_\theta(e)|_{\Hn}=\theta I$. If $\lambda_j(A)$ are the non-Perron eigenvalues, the original-dynamics and frozen-response multipliers are respectively
\begin{equation}\label{eq:two-spectra}
 \mu_j^\Phi=(a+b\theta)\lambda_j(A),\qquad
 \mu_j^T=\theta\frac{b\lambda_j(A)}{1-a\lambda_j(A)}.
\end{equation}
They need not cross the unit circle at the same parameters. The relation~\eqref{eq:affine-multiplier} used in the uniform retention family is special to its scalar tangent action and does not justify using $T$ as the original dynamics on general networks.

At a stationary pair $(y,\gamma)$, the fraction of the stock at node $i$ attributable to injection at node $j$ under the discounted response is
\begin{equation}\label{eq:reverse}
 K_{ij}=\frac{Q_{ij}\gamma_j}{y_i}.
\end{equation}
The reverse response kernel $K$ is row-stochastic, with $K\1=\1$ and $y^\top K=\gamma^\top$. For a relative perturbation $h=\diag(\gamma)\xi$ with $\gamma^\top\xi=0$,
\begin{equation}\label{eq:reverse-derivative}
 DT(\gamma)h
 =\theta\diag(\gamma)(I-\1\gamma^\top)K\xi.
\end{equation}
The factor $I-\1\gamma^\top$ subtracts the $\gamma$-weighted mean of $K\xi$ from every coordinate. Relative perturbations that differ by a constant represent the same projective direction, so this subtraction acts as the identity on the quotient by constant vectors. The nontrivial local eigenvalues of $T$ are therefore $\theta$ times the non-Perron eigenvalues of $K$. The row-variation coefficient of $K$ is
\begin{equation}
 c_{\mathrm{row}}(K)=\frac12\max_{i,j}\sum_\ell|K_{i\ell}-K_{j\ell}|
\end{equation}
and satisfies $\osc(K\xi)\le c_{\mathrm{row}}(K)\osc(\xi)$, where $\osc(\xi)=\max_i\xi_i-\min_i\xi_i$. Combined with~\eqref{eq:reverse-derivative}, the coefficient bounds the local contraction of the frozen-response iteration.

\section{Heavy-tail moment bounds and negative powers}\label{app:tails}
The largest coordinates of a heavy-tailed profile can dominate its power-rule allocation, but their contribution depends on how the upper cutoff grows. We calculate the benchmark IPR for a prescribed profile with a growing cutoff. Appendix~\ref{app:selection} gives sufficient conditions for transferring these orders to nonlinear stationary injection.

\subsection{Positive powers and the upper cutoff}
Here $v^{(n)}$ is a specified probability vector and $\widehat v_i=nv_i$. Let $N_n(x)=\#\{i:\widehat v_i>x\}$. Assume that for some $\tau\in(1,2)$, cutoff exponent $\zeta>0$, and positive constants independent of $n$,
\begin{equation}\label{eq:tail-assumption}
 c_1nx^{-\tau}\le N_n(x)\le c_2nx^{-\tau}
 \quad (x_0\le x\le c_0n^\zeta),\qquad
 \max_i\widehat v_i\le C_0n^\zeta.
\end{equation}
The associated density exponent is $\alpha=\tau+1$. The bounds constrain the tail count uniformly over the stated range for each deterministic profile. An arbitrary random sample need not satisfy them. They imply $\zeta\le1/\tau$.

\begin{proposition}[Benchmark concentration orders]\label{prop:tails}
Under~\eqref{eq:tail-assumption}, for every fixed $\theta>0$,
\begin{equation}\label{eq:five-regimes}
 \IPR(P_\theta(v^{(n)}))\asymp
 \begin{cases}
 n^{-1},&0<\theta<\tau/2,\\
 n^{-1}\log n,&\theta=\tau/2,\\
 n^{-1+\zeta(2\theta-\tau)},&\tau/2<\theta<\tau,\\
 n^{-1+\zeta\tau}(\log n)^{-2},&\theta=\tau,\\
 n^{-1+\zeta\tau},&\theta>\tau.
 \end{cases}
\end{equation}
\end{proposition}
\proofref{app:proof-tails}

For the natural cutoff $\zeta=1/\tau$, the benchmark localizes exactly when $\theta>\tau$. At $\theta=\tau$, its IPR still vanishes. For a lower cutoff $\zeta<1/\tau$, the benchmark IPR vanishes for every fixed $\theta>0$.
The IPR is the sum of squared power-rule weights divided by the square of their sum. The sum of squared weights changes order at $\theta=\tau/2$, while the sum of weights changes order at $\theta=\tau$. These moment thresholds also govern normalized heavy-tailed weights~\cite{gradenigobertin2017}; the explicit cutoff assumption determines whether the resulting IPR vanishes. The two boundary cases retain logarithmic factors, so a zero power-law exponent alone does not establish localization.

\subsection{Negative powers and the lower tail}
Let $v$ be strictly positive. Negative exponents put the largest weights on the smallest coordinates of $v$. If at least $cn$ coordinates are at most $C\min_i v_i$, with fixed $c>0$ and $C\ge1$, then every fixed $\theta<0$ satisfies
\begin{equation}\label{eq:negative-IPR}
 \frac1n\le\IPR(P_\theta(v))\le\frac{C^{|\theta|}}{cn}.
\end{equation}
\proofref{app:proof-negative}

Under this condition, a positive fraction of nodes have power-rule weights comparable to the largest weight. Their combined contribution prevents any one node from receiving a nonvanishing benchmark share. Information about the upper heavy tail alone does not constrain this lower end of the profile.

\subsection{A deterministic finite-size benchmark check}
The accompanying calculation uses $\tau=1.5$ and
\begin{equation}\label{eq:numeric-quantile}
 \widetilde v_i=\min\{(n/i)^{1/\tau},n^\zeta\},\qquad
 v_i=\widetilde v_i/\sum_k\widetilde v_k,
\end{equation}
with $\zeta=1/\tau$ for the natural cutoff and $\zeta=1/2$ for the structural cutoff. We evaluate the benchmark at $n=250$, $1000$, $4000$, $16000$, $64000$ and $256000$ for fixed exponents spanning both thresholds. Figure~\ref{fig:benchmark} compares the fitted decay exponents (the negatives of the fitted log--log slopes) with the limiting decay exponents in~\eqref{eq:five-regimes}. The profiles are deterministic, so these fitted exponents have no sampling uncertainty. Their finite-size error still reflects the logarithmic corrections near the thresholds. A fitted decay exponent close to zero therefore does not by itself establish localization.

\begin{figure}[!htbp]
\centering
\includegraphics[width=\linewidth]{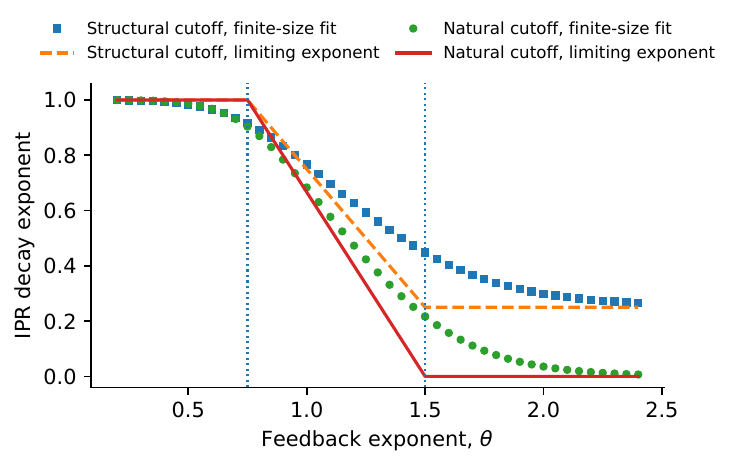}
\caption{Deterministic benchmark concentration with natural and structural cutoffs. Markers are negative log--log slopes fitted over the six stated sizes. Lines are the limiting power-law decay exponents. Vertical dotted lines mark $\theta=\tau/2$ and $\theta=\tau$. Logarithmic corrections at the thresholds are not captured by the limiting power-law exponent alone.}\label{fig:benchmark}
\figurealt{Markers show fitted finite-size decay exponents and lines show the theoretical power-law exponents under natural and structural cutoffs. The largest discrepancies occur near the two moment thresholds, where logarithmic factors affect the fitted slopes.}
\end{figure}

\section{Allocation bias and degree-based reductions}\label{app:reductions}
The power rule can favour nodes with larger or smaller stocks while leaving injection diffuse. We examine this allocation bias near uniform injection and in degree-based reductions. Nodes within a degree class can still have different injection shares, so class averages alone do not determine the node-level IPR.

\subsection{Weak feedback at fixed dimension}
At $\theta=0$, the stationary injection is $e=\1/n$ and its stock is $Qe$. The implicit-function theorem applies because the derivative of $T$ with respect to $\gamma$ is zero at $\theta=0$. Differentiating with respect to $\theta$ gives
\begin{equation}\label{eq:weak-feedback}
 \gamma_i(\theta)=\frac1n+
 \frac\theta n\left[\log(Qe)_i-\frac1n\sum_j\log(Qe)_j\right]
 +O(\theta^2).
\end{equation}
The $O(\theta^2)$ remainder in~\eqref{eq:weak-feedback} concerns $\theta\to0$ at fixed dimension and fixed positive $Q$; its constant need not be uniform as the network grows. Weak positive feedback favours nodes whose response to uniform injection exceeds the geometric mean of the coordinates of $Qe$. Weak negative feedback favours nodes whose response falls below that mean. Degree alone need not order this response on a general network.

\subsection{A two-block one-step response}
For a random walk on an undirected network, consider a two-block reduction of one-step transport with representative degrees $k_H>k_P>0$. Let $\omega_{ab}$ be the probability that a neighbour of a group-$a$ node belongs to group $b$, with all four probabilities strictly positive and each row summing to one. Write $\chi$ for the ratio of per-node injection in the hub group to per-node injection in the peripheral group. The ratio of one-step responses is
\begin{equation}\label{eq:two-block}
 r(\chi)=\frac{\omega_{HH}\chi+(k_H/k_P)\omega_{HP}}
 {(k_P/k_H)\omega_{PH}\chi+\omega_{PP}},\qquad
 \chi=r(\chi)^\theta.
\end{equation}
The two-block equation uses only one transport step. At finite forcing, the full stationary response includes all discounted transport paths through $Q$. The same reduced analysis applies to any prescribed positive fractional-linear response ratio.

\begin{proposition}[Sign and bounds for a reduced allocation bias]\label{prop:block}
Let $G=r(1)>1$ and $|\theta|<1$. Equation~\eqref{eq:two-block} has a unique positive solution. Its logarithm has the sign of $\theta$, and
\begin{equation}\label{eq:block-bound}
 \frac{|\theta|\log G}{1+|\theta|}
 \le|\log\chi|\le
 \frac{|\theta|\log G}{1-|\theta|}.
\end{equation}
\end{proposition}
\proofref{app:proof-block}

The reduced inflow advantage $G=r(1)$ compares the groups' responses under equal per-node injection. Within the regime $|\theta|<1$ of Proposition~\ref{prop:block}, $G>1$ implies $\chi>1$ for $0<\theta<1$ and $\chi<1$ for $-1<\theta<0$. Thus positive feedback favours the hub group and negative feedback favours the peripheral group in this reduced model; zero feedback gives equal per-node shares. These comparisons do not determine how a particular rewiring changes $G$, or whether injection within the favoured group localizes on a few nodes.

\subsection{Discounted degree-class response}
For an undirected network, let $P(k)$ be the fraction of nodes of degree $k$, and let $\omega(k\to k')$ be the probability that a neighbour of a degree-$k$ node has degree $k'$. We include only degree classes with $P(k)>0$. In the annealed approximation, we replace each node's neighbourhood by this conditional degree distribution. The per-node class-response operator is
\begin{equation}
 B_{kk'}=\frac{k}{k'}\omega(k\to k').
\end{equation}
Counting edges between degree classes from either end gives
\begin{equation}
 kP(k)\omega(k\to k')=k'P(k')\omega(k'\to k),
 \qquad \sum_kP(k)B_{kk'}=P(k').
\end{equation}
Let $\eta(k)/n$ be the per-node injection share in degree class $k$, so $\sum_kP(k)\eta(k)=1$. The discounted class response is
\begin{equation}\label{eq:class-response}
 \mathcal Q=bB(I-aB)^{-1},\qquad
 \eta(k)=\frac{[\mathcal Q\eta](k)^\theta}
 {\sum_{k'}P(k')[\mathcal Q\eta](k')^\theta}.
\end{equation}
The operator $\mathcal Q$ averages the class responses over all discounted transport paths. A one-step approximation using $B$ generally omits the contributions from longer paths.

Write $\langle k\rangle=\sum_k kP(k)$ for the mean degree. For uncorrelated mixing, $\omega(k\to k')=k'P(k')/\langle k\rangle$ and $B_{kk'}=kP(k')/\langle k\rangle$. Hence $B^2=B$ and $\mathcal Q=B$ at every $a,b$. The class fixed point is exactly
\begin{equation}\label{eq:uncorrelated-class}
 \eta(k)=\frac{k^\theta}{\sum_{k'}P(k')k'^\theta}.
\end{equation}
Equation~\eqref{eq:uncorrelated-class} is exact for the annealed operator. On an individual network, nodes in the same degree class can have different responses, and those differences can dominate the node-level IPR.

Supplementary Material, Section S3 gives a residual test for proposed degree-response laws and an assortativity comparison for the discounted two-block model. Section S4 states a sufficient condition for monotone class trajectories and distinguishes it from convergence alone.

\subsection{Range bounds for prescribed responses}
For a positive prescribed vector $v$, let $R(v)=\max_i v_i/\min_i v_i$. For every real $\theta$,
\begin{equation}
 \max_iP_\theta(v)_i\le\frac{R(v)^{|\theta|}}n,
 \qquad \IPR(P_\theta(v))\le\frac{R(v)^{|\theta|}}n.
\end{equation}
The first inequality compares the largest weight with $n$ times the smallest, and the second uses~\eqref{eq:ipr-bounds}. At fixed dimension, $P_\theta(v)$ tends as $\theta\to+\infty$ to the uniform distribution on the maximizers of $v$, and as $\theta\to-\infty$ to that on its minimizers. These fixed-profile limits do not identify the limiting supports of a $\theta$-dependent nonlinear equilibrium.

\section{Numerical methods and validation}\label{app:numerical}
\subsection{Initial stocks and sparse-network realizations}\label{app:sparse-methods}
The mathematical node indices run from one to $n$; node indices in the CSV files are zero-based. The smooth-background experiments use a background start $x_0=v_n$ and three seeded starts $x_0=(1-q_n)v_n+q_ne_j$, with $j=1$, $1+\lfloor n/3\rfloor$ and $1+\lfloor2n/3\rfloor$. The nonquadratic heavy-tail experiment uses the same stock formula at $j=1$, $1+\lfloor n/2\rfloor$ and $n$, corresponding to the highest, middle-ranked and lowest background weights. The slow heavy-tail sequence uses only $j=1$ and $n$, in addition to the background start.

For the sparse experiment, NetworkX generates five simple 20-regular graphs at each $n\in\{250,1000,4000\}$. We check that each graph is connected. Replicate $r\in\{0,1,2,3,4\}$ uses seed $857+n+100003r$. The near-uniform start is the normalized vector with entries $1+0.01(U_i-1/2)$, where the $U_i$ are independent draws from the uniform distribution on $[0,1)$, using NumPy's default generator with seed one larger than the graph seed. The concentrated starts are $x_0=0.01\1/n+0.99e_j$ at the same three node indices as in the smooth-background experiment. Each graph is held fixed during all four trajectories. The package retains every adjacency matrix, so the reported graph realizations can also be used without regenerating them.

\subsection{Construction of the profile-preserving perturbations}\label{app:multimode}\label{app:perturb-modes}
The single- and two-mode calculations use the background, retention and sizes from Section~\ref{sec:perturb-experiment}. With $\varphi_i=2\pi(i-1)/n$ for $i=1,\ldots,n$, take the raw modes
\begin{align}
 \widetilde u_{1,i}&=\sin(3\varphi_i)+0.25\cos\varphi_i,&
 \widetilde w_{1,i}&=\cos(2\varphi_i)+0.3\sin\varphi_i,\label{eq:mode-one}\\
 \widetilde u_{2,i}&=\cos(5\varphi_i)-0.2\sin(2\varphi_i),&
 \widetilde w_{2,i}&=\sin(4\varphi_i)+0.2\cos(3\varphi_i).\label{eq:mode-two}
\end{align}
We subtract each raw mode's $v_n$-weighted mean and rescale the centered vector so that its largest absolute coordinate is one. The single-mode experiment uses $u_1,w_1$ in~\eqref{eq:numeric-perturb}. The two-mode experiment uses $c_1=0.4$, $c_2=0.35$ and
\begin{equation}\label{eq:multimode}
 F_{ij}=q_nv_{n,i}\sum_{\ell=1}^2 c_\ell u_{\ell,i}w_{\ell,j}.
\end{equation}
The weighted centering gives $\1^\top F=0$ and $Fv_n=0$, while $|F_{ij}|\le0.75q_nv_{n,i}$. The bound $1-\rho_n-0.75q_n>0$ verifies positivity at the reported sizes. The transport action is evaluated from its diagonal and rank-one terms, without storing a dense matrix.

\subsection{Original dynamics, stopping rules and residuals}
All dynamical experiments iterate the original normalized stock map~\eqref{eq:physical}. The smooth-background and sparse runs stop when successive stocks differ by less than $2\times10^{-13}$ in $\ell^1$, with a cap of 30,000 updates. The heavy-tail runs use $10^{-13}$ and a cap of 10,000 updates. A run that reaches its cap raises an error rather than being recorded as converged.

To evaluate positive powers, the smooth-background and sparse routines first divide the stock by its largest coordinate. The heavy-tail routine instead subtracts the largest log coordinate before exponentiation. These rescalings leave $P_\theta$ unchanged and reduce the risk of numerical underflow.

For the final computed stock $x$, we set $\gamma=P_\theta(x)$ and recompute the original-dynamics residual $\|\Phi(x)-x\|_1$. We also evaluate the frozen-equation residual $\|P_\theta(Q\gamma)-\gamma\|_1$. In the smooth-background and sparse runs, we independently evaluate $Q\gamma$ through the Neumann series for $Q$, stopping when the remaining total weight is below $2\times10^{-15}$. In the exact heavy-tail family, we use $Q\gamma=(1-q)v+q\gamma$.

\subsection{Identifying the computed profiles}
Different initial stocks can reach the same stationary profile. We therefore identify outcomes from the final injections, rather than from the labels of their starting stocks. We record the node with the largest injection share, the share at each seeded node, and pairwise distances between final injection profiles. The validation uses these quantities to count distinct outcomes. Representative final stock and injection vectors are retained, together with every sparse adjacency matrix. In tables, $y$ denotes the numerical approximation to a stationary stock.

\subsection{Tangent norms and eigenvalue diagnostics}
For the exact retention family and its profile-preserving perturbations, we bound the induced $\ell^1$ norm of $D\Phi(x)$ on $\Hn$ at the final stock $x$. A bound below one supports local attraction without requiring estimates of nearly degenerate eigenvalues. A bound above one is inconclusive. We calculate the bound from the exact tangent norm of the power-rule derivative.
\begin{lemma}[Exact tangent $\ell^1$ norm of the power derivative]\label{lem:DPnorm}
Let $n\ge2$, $w\in\R_{>0}^n$, $\theta>0$, $t=P_\theta(w)$ and $r_i=t_i/w_i$. Choose $j$ with largest $r_j$, and let $r_{(2)}=\max_{i\ne j}r_i$. Then
\begin{equation}\label{eq:DPnorm}
 \|DP_\theta(w)|_{\Hn}\|_{1\to1}
 =\theta\{r_j(1-t_j)+r_{(2)}t_j\}.
\end{equation}
\end{lemma}
\proofref{app:proof-DPnorm}

When the largest $r_i$ is attained at more than one node, any maximizing index $j$ gives the same norm.

For the rank-one perturbation~\eqref{eq:numeric-perturb},
\begin{equation}
 \|F|_{\Hn}\|_{1\to1}
 =c q\left(\sum_i v_i|u_i|\right)\frac{\max_i w_i-\min_i w_i}{2}.
\end{equation}
Thus the saved derivative diagnostic is the upper bound
\begin{equation}\label{eq:numeric-derivative-bound}
 \|D\Phi(x)|_{\Hn}\|_{1\to1}
 \le\bigl(\rho+\|F|_{\Hn}\|_{1\to1}\bigr)
       \bigl(a+b\|DP_\theta(x)|_{\Hn}\|_{1\to1}\bigr).
\end{equation}
For the two-mode perturbation~\eqref{eq:multimode}, we use
\begin{equation}\label{eq:multimode-norm}
 \|F|_{\Hn}\|_{1\to1}
 \le\frac{q_n}{2}\sum_{\ell=1}^2c_\ell
       \left(\sum_i v_{n,i}|u_{\ell,i}|\right)
       \left(\max_jw_{\ell,j}-\min_jw_{\ell,j}\right).
\end{equation}
Combining~\eqref{eq:multimode-norm} with~\eqref{eq:numeric-derivative-bound} gives the saved bound on the original-dynamics derivative.

For sparse runs from concentrated starts, we apply $D\Phi$ as a matrix-free linear operator, projecting both its input and output onto $\Hn$ by subtracting their means. ARPACK returns four estimated eigenpairs of largest modulus. For each run, we save the largest eigenvalue modulus and the largest Ritz residual among those four pairs, rather than the individual pairs.

For near-uniform starts, we evaluate stability at the uniform state using~\eqref{eq:two-spectra}. A symmetric eigensolver estimates the two largest eigenvalues of $A$. After excluding the Perron eigenvalue, we multiply the remaining estimate by $a+b\theta=1.1$. The holding probability in~\eqref{eq:sparse-numeric-A} makes every eigenvalue of $A$ positive, so the largest non-Perron eigenvalue also has the largest non-Perron modulus. For these near-uniform runs, the saved Ritz-residual field is zero because the four-eigenpair calculation was not used. The zero is a placeholder, not a measured residual.

\subsection{Outcomes under profile-preserving perturbations}\label{app:perturb-outcomes}
The single-mode experiment contains 32 trajectories. For the background start, the maximum coordinatewise relative error of injection from the benchmark falls from about $0.181$ at $n=500$ to $0.029$ at $n=32000$. The seeded injection IPRs lie between $0.889$ and $0.937$, and every seeded outcome has its largest share at the intended node. All reported tangent-norm bounds are below $0.625$. Original-dynamics residuals are below $8.4\times10^{-14}$ in $\ell^1$, and frozen-equation residuals are below $2.5\times10^{-13}$.

The two-mode calculation reaches four distinct profiles at each size, again with the largest seeded share at the intended node. The dataset contains all 16 runs and the 24 within-size pairwise injection distances. The tangent-norm bounds are below one at every reported stock.

\subsection{Finite-size diagnostics and a slow heavy-tail sequence}\label{app:slow}\label{app:heavy-scales}
Table~\ref{tab:heavy-scales} gives the proof parameters for the nonquadratic experiment. At $n=10^5$, the tangent-norm bound exceeds one for two outcomes even though the iteration converges. These upper bounds do not establish instability of either outcome. Table~\ref{tab:heavy-outcomes} reports $n=10^6$, where all four bounds are below one.

% Generated by code/make_tables.py. Do not edit numerical entries by hand.
\begin{table}[!htbp]
\centering
\caption{Proof parameters for the heavy-tail sequence with $\tau=1.95$, $\theta=1.4$ and $q=1.1q_{\mathrm{bal}}$. Here $\varepsilon=Z/q^\theta$, $\kappa=qv_{\max}^{\theta-1}/Z$ and $\delta=\varepsilon\kappa$. Small $\varepsilon$ supports the condensed branches, while small $\kappa$ controls feedback near the diffuse branch.}\label{tab:heavy-scales}
\small
\begin{tabular}{rrrrr}
\toprule
$n$ & $q$ & $\varepsilon$ & $\kappa$ & $\delta$ \\
\midrule
$10\,000$ & 0.144128 & 0.464616 & 0.584032 & 0.271351 \\
$100\,000$ & 0.082332 & 0.414771 & 0.521375 & 0.216251 \\
$1\,000\,000$ & 0.046607 & 0.371225 & 0.466637 & 0.173227 \\
\bottomrule
\end{tabular}
\end{table}

For the slower sequence, $\tau=1.5$, $\theta=1.2$, $a=b=1/2$ and $q_n=n^{-0.15}$. The exponent satisfies~\eqref{eq:sigma}. At each size, the original map is iterated from $v_n$ and from stocks seeded at the nodes with the highest and lowest background weights. Table~\ref{tab:heavy-slow} reports the background-start outcome. The three final injection profiles differ by less than $3\times10^{-12}$ in pairwise $\ell^1$ distance at every size. Their original-dynamics residuals are below $10^{-13}$. The background and two seeded starts therefore did not reveal distinct stationary profiles; other initial stocks may still reach different outcomes.

% Generated by code/make_tables.py. Do not edit numerical entries by hand.
\begin{table}[!htbp]
\centering
\caption{A slow heavy-tail sequence with $\tau=1.5$, $\theta=1.2$ and $q=n^{-0.15}$. The background-start outcome is shown. At each size, the highest-weight and lowest-weight starts reach the same profile within numerical tolerance.}\label{tab:heavy-slow}
\small
\begin{tabular}{rrrrr}
\toprule
$n$ & $q$ & $\varepsilon$ & $\kappa$ & $\max_i\gamma_i$ \\
\midrule
$1\,000$ & 0.354813 & 1.002764 & 0.632053 & 0.131293 \\
$32\,000$ & 0.210973 & 0.996694 & 0.553459 & 0.057270 \\
$1\,000\,000$ & 0.125893 & 0.967535 & 0.500750 & 0.027101 \\
\bottomrule
\end{tabular}
\end{table}

\FloatBarrier
\subsection{Sparse-network outcomes}\label{app:sparse-outcomes}
For the sparse transport in~\eqref{eq:sparse-numeric-A}, all 15 near-uniform starts reach the uniform state and all 45 concentrated starts reach concentrated profiles. Table~\ref{tab:sparse} and Figure~\ref{fig:sparse} summarize the concentrated-start runs. The reported ranges describe variation across five saved graph realizations and three concentrated starts per graph at each size. They are not confidence intervals.

% Generated by code/make_tables.py. Do not edit numerical entries by hand.
\begin{table}[!htbp]
\centering
\caption{Sparse 20-regular networks with holding probability $0.8$, $a=b=1/2$ and $\theta=1.2$. Each row summarizes three concentrated starts on each of five graph realizations. The IPR columns are means over the 15 runs. The last column gives the largest modulus among the four tangent eigenvalues estimated in each run, maximized over all 15 runs. Figure~\ref{fig:sparse} shows the IPR ranges.}\label{tab:sparse}
\small
\begin{tabular}{rrrr}
\toprule
$n$ & Mean $\IPR(\gamma)$ & Mean $\IPR(y)$ & $\max|\widehat\mu|$ \\
\midrule
$250$ & 0.350983 & 0.166278 & 0.825143 \\
$1\,000$ & 0.392644 & 0.184633 & 0.809486 \\
$4\,000$ & 0.407286 & 0.191157 & 0.803830 \\
\bottomrule
\end{tabular}
\end{table}

The holding probability gives $Q\ge(2/3)I$. Theorem~\ref{thm:return} therefore gives an injection-IPR lower bound of $(2/3)^{11}$, while Proposition~\ref{prop:return-stock} requires a stock-to-injection IPR ratio of at least $4/9$. Every reported concentrated-start run respects the stock bound. The iteration does not perform a global variational optimization, so the observed injections need not be the maximizers used in the existence proof.

Original-dynamics residuals are below $2.0\times10^{-13}$, and frozen-equation residuals are below $4.2\times10^{-13}$. At the concentrated-start outcomes, the largest modulus among the four tangent eigenvalues estimated by ARPACK is below $0.826$ in every run, and the largest Ritz residual is below $4.2\times10^{-11}$. The uniform-state multiplier estimates are also below one for all 15 graphs. Together with convergence from the tested starts, these diagnostics support coexistence of uniform and concentrated attractors on the reported finite networks. The eigenvalue estimates do not certify the full spectrum, and the tested starts do not establish selection from arbitrary initial stocks.

\begin{figure}[!htbp]
\centering
\includegraphics[width=\linewidth]{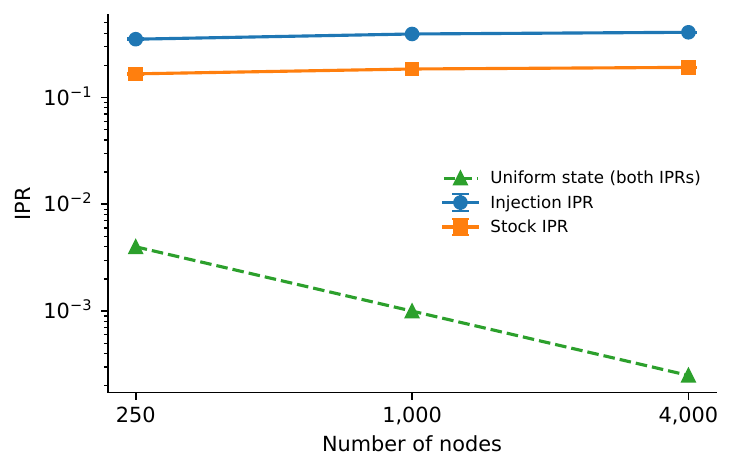}
\caption{Sparse-network outcomes under~\eqref{eq:sparse-numeric-A}. At each size, points give the mean injection and stock IPRs from three concentrated starts on each of five graphs. Bars give the ranges over these 15 runs. The uniform-state curve is $1/n$.}\label{fig:sparse}
\figurealt{Injection and stock IPRs from concentrated starts stay well above the falling uniform-state curve. Vertical bars show the range across graphs and seeded starts, rather than statistical confidence intervals.}
\end{figure}

\FloatBarrier
\subsection{Implementation checks and reproducibility}
The validation script checks the discounted-response identities, analytic derivatives against finite differences, the tangent-norm formula against column differences and the perturbation resolvent bound on small dense matrices. For selected dimensions and retention values, it enumerates all two-level quadratic equilibria and compares the locally attracting count with Theorem~\ref{thm:classification}. It also checks the low-rank nonquadratic update against a dense-matrix recursion and verifies the balanced-retention identities.

The saved datasets are checked for normalization, positivity, stationary residuals, branch identities and stock-IPR bounds. Figures and tables are generated from those datasets. The accompanying code and environment records allow all calculations to be reproduced; the supplied figures and tables also allow a paper-only build.

All numerical calculations use double precision, without interval enclosures. Small residuals and the derivative diagnostics support the reported finite-size outcomes; they do not by themselves prove the existence or stability of an exact nearby equilibrium. The validation also does not verify the asymptotic proofs, enumerate uncomputed equilibria or establish the full sparse spectrum.

\bibliographystyle{unsrtnat}
\bibliography{references}
\end{document}